\documentclass[runningheads]{llncs}
\usepackage{graphicx}
\usepackage{cite}
\usepackage{amsmath,amssymb,amsfonts}
\usepackage{mathtools}
\usepackage{algorithmic}
\usepackage{graphicx}
\usepackage{textcomp}
\usepackage{tikz}
\usepackage{caption}
\usepackage{cuted}
\usepackage[utf8]{inputenc}
\usepackage{pgfplots} 
\usepackage{pgfgantt}
\usepackage{pdflscape}
\usepackage{appendix}
\usepackage{pst-plot}
\usetikzlibrary{spy}
\usetikzlibrary{positioning,calc}
\usetikzlibrary{decorations.pathmorphing,calc,shapes,shapes.geometric,patterns}
\usetikzlibrary{shapes.multipart}
\usepackage{xfrac}
\usepackage{colortbl}
\usepackage{cancel}
\usepackage{comment}
\usetikzlibrary{arrows,positioning,calc,intersections}
\usetikzlibrary{datavisualization.formats.functions}
\def\BibTeX{{\rm B\kern-.05em{\sc i\kern-.025em b}\kern-.08em
    T\kern-.1667em\lower.7ex\hbox{E}\kern-.125emX}}

\usepackage{pgfplots}
\usepgfplotslibrary{fillbetween}
\usetikzlibrary{arrows, decorations.markings}
\usepackage{algorithm}
\usepackage{algorithmic}

\newcommand{\setx}{\ensuremath{\mathcal{X}}}
\newcommand{\sety}{\ensuremath{\mathcal{Y}}}
\newcommand{\setz}{\ensuremath{\mathcal{Z}}}
\newcommand{\setd}{\ensuremath{\mathcal{D}}}

\newcommand{\sete}{\ensuremath{\mathcal{E}}}
\newcommand{\bs}[1]{\boldsymbol{#1}}
\newcommand{\Wg}{W_{\text{g}}}
\newcommand{\Vg}{V_{\text{g}'}}

\newcommand{\up}{u^\prime}
\newcommand{\upp}{u^{\prime\prime}}
\newcommand{\Mp}{M^\prime}
\newcommand{\Mpp}{M^{\prime\prime}}
\newcommand{\dc}{\mathcal{D}}
\newcommand{\dcp}{\mathcal{D}^\prime}
\newcommand{\dcpp}{\mathcal{D}^{\prime\prime}}
\newcommand{\C}{\mathcal{C}}
\newcommand{\Cp}{\mathcal{C}^\prime}
\newcommand{\Cpp}{\mathcal{C}^{\prime\prime}}
\newcommand{\Cn}{\mathbb{C}^{N_T}}
\newcommand{\Cr}{\mathbb{C}^{N_R}}

\newcommand{\h}{\mathsf{H}}

\newcommand{\mbf}{\mathbf}
\newcommand{\mc}{\mathcal}
\newcommand{\mbb}{\mathbb}
\newcommand{\Zp}{Z^\prime}

\newcommand{\circlearrow}{}% just in case
\DeclareRobustCommand{\circlearrow}{%
  \mathrel{\vphantom{\rightarrow}\mathpalette\circle@arrow\relax}%
}
\newcommand{\circle@arrow}[2]{%
  \m@th
  \ooalign{%
    \hidewidth$#1\circ\mkern1mu$\hidewidth\cr
    $#1-$\cr}%
}

\begin{document}
\title{Optimal Signal Processing for Common Randomness Generation over MIMO Gaussian Channels with Applications in Identification}
\titlerunning{Optimal Signal Processing for CR Generation over Gaussian Channels}
\author{Rami Ezzine \inst{1,3}\orcidID{0000-0002-3432-4447} \and Wafa Labidi\inst{1,2,3}\orcidID{0000-0001-5704-1725} \and
Christian Deppe\inst{2,3}\orcidID{0000-0002-2265-4887} \and
Holger Boche\inst{1,3,4,5}\orcidID{0000-0002-8375-8946}}
\authorrunning{R. Ezzine, W. Labidi, C. Deppe, H. Boche}
\institute{Technical University of Munich, TUM School of Computation, Information and Technology, Munich, Germany \and
Technical University of Braunschweig, Institute for Communications Technology, Braunschweig, Germany \and 6G-life, 6G research hub, Germany
\and  {\color{black}{Munich Center for Quantum Science and Technology, Munich, Germany}} \and {\color{black}{Munich Quantum Valley, Munich, Germany}} \\
\email{rami.ezzine@tum.de, wafa.labidi@tum.de, christian.deppe@tu-braunschweig.de, boche@tum.de}}
\maketitle              % typeset the header of the contribution
\centerline{\bf In memory of Ning Cai}
\begin{abstract}
Common randomness (CR), as a resource, is not commonly exploited in existing practical communication systems. In the CR generation framework, both the sender and receiver aim to generate a common random variable observable to both, ideally with low error probability. The availability of this CR allows us to implement correlated random protocols that can lead to faster and more efficient algorithms. Previous work focused on CR generation over perfect channels with limited capacity. In our work, we consider the problem of CR generation from independent and identically distributed (i.i.d.) samples of a correlated finite source with one-way communication over a Gaussian channel.  We first derive the CR capacity for single-input single-output (SISO) Gaussian channels. This result is then used for the derivation of the CR capacity in the multiple-input multiple-output (MIMO) case. CR plays a key role in the identification scheme since it may allow a significant increase in the identification capacity of channels. In the identification framework, the decoder is interested in knowing \emph{whether} a specific message of special interest to him has been sent or not, rather than knowing \emph{what} the received message is. In many new applications, such as several machine-to-machine and human-to-machine systems and the tactile internet, this post-Shannon scheme is more efficient than classical transmission. In our work, we also consider a CR-assisted secure identification scheme and develop a lower bound on the corresponding secure identification capacity.

\end{abstract}

\section{Introduction}

In the common randomness (CR) generation framework, the communicating parties, often referred to as terminals, aim to generate a common random variable observable to both, ideally with low error probability \cite{part2}. The availability of this shared randomness enables the implementation of correlated random protocols, which can result in faster and more efficient algorithms \cite{corrsurvey}\cite{Naor2020}.

CR is considered a highly promising resource for future communication systems due to its essential role in various communication tasks.
For instance, CR plays a key role in the identification scheme, an approach in communications developed by Ahlswede and Dueck \cite{Idchannels}.  Interestingly, CR can significantly increase the identification capacity of channels. As a result, an enormous performance gain can be achieved by taking advantage of this resource. In the identification framework, the encoder sends an identification message over the channel. In contrast to transmission \cite{shannon}, the decoder is now interested in knowing whether a specific message of special interest to him was sent or not, rather than knowing what the received message is.
In many new applications such as several machine-to-machine and human-to-machine systems \cite{application}, industry 4.0 \cite{industrie40}, 6G communication systems\cite{6Gcomm}\cite{6Gpostshannon} and digital watermarking\cite{Moulin,watermarkingahlswede,watermarking}, it appears that the identification scheme is more efficient than the classical transmission scheme. 

CR is perhaps more evident in cryptography. In fact, it is used in the secret key generation problem \cite{part1}. Note that the key generation problem is an example of common randomness generation where secure communication between sender and receiver is ensured. It is worth mentioning that an interesting scenario in this context is the use of WiFi to exploit common randomness as a key, as introduced in \cite{Wifi}. In our work, however, we will not impose any secrecy constraints.

Additionally, CR plays an important role in modular coding schemes for secure communication. As discussed in \cite{semanticwiretap}, modular schemes for semantic security have been designed to integrate with arbitrary error-correcting codes, thereby establishing semantic security. Often, in seeded modular coding scenarios, legitimate parties possess CR as an additional resource, which can be used as a seed \cite{semanticsecurity}.

Furthermore, it was demonstrated in \cite{6Gandtrustworthiness} that exploiting CR as a resource facilitates state estimation with error-free reconstruction of the state distribution in joint sensing and communication applications. Moreover, it was established that the presence of this resource is crucial for the perfect reconstruction of the state distribution. This characteristic of CR is highly intriguing for joint sensing and communication applications in 6G \cite{6Gcomm}.\color{black}

The Post-Shannon resource of CR can also be leveraged to achieve inherent resilience for the tactile internet and quantum communication systems. Specifically, when legitimate parties have access to a common random source as an additional coordination resource, communication becomes resilient against denial-of-service (DOS) attacks by jammers. Remarkably, only a few bits of CR are needed to counteract the jamming attack \cite{jammerref}. \color{black}Incorporating resilience by design is crucial for ensuring trustworthiness in 6G \cite{trustworthiness}. \color{black} 
In \cite{6Gperspective}, CR was named as an important additional resource for future 6G systems due to the aforementioned potential for a wide range of applications. The first network operators are already starting to set up a research network infrastructure for the generation and distribution of CR.

Previous work in \cite{part2} focused on the problem of CR generation from finite sources with unidirectional communication over perfect rate-limited channels. In our work, we consider the case when the terminals communicate over single-input single-output (SISO) as well as multiple-input multiple-output (MIMO) Gaussian channels. Gaussian channels are well-known for their practical relevance in many communication situations, e.g., satellite and deep space communication links \cite{mobile}, wired and wireless communications, etc. We characterize the CR capacity for our specified model. The latter is defined as the maximum rate of CR one can achieve using the resources available in the model.

In our work, we also address the problem of secure identification over Gaussian wiretap channels (GWCs) with common randomness (CR) being available as a resource. Secure identification has been extensively studied for discrete alphabets \cite{Boche2019SecureIU,IDforDataStorage,application} over recent decades due to its important potential use in many future scenarios. Indeed, for discrete channels, it was proved in \cite{Boche2019SecureIU} that secure identification is robust under channel uncertainty and against jamming attacks. It has been demonstrated that, in contrast to secure transmission, the identification capacity of the discrete wiretap channel coincides with the capacity of the main channel. This holds true only if the secrecy capacity elaborated in \cite{sectransmission} is strictly positive. Recently, the results were extended to the Gaussian case in \cite{wafapaper}. However, as far as we know, there has been limited research on the secure CR-assisted identification capacity for GWCs. The wiretap channel is a basic model considered by Wyner \cite{Wyner} in information-theoretic security. The wiretapper, in contrast to the discrete case, is now not limited anymore and has an infinite alphabet. Moreover, we assume the wiretapper has access to the correlated source signals. This is advantageous for him because he has no limitations on the hardware resolution.
In our coding scheme, the sender wants to send a secure identification message to the legitimate receiver so that the receiver is able to identify his message. Both the sender and the receiver share an extra resource of randomness. Meanwhile, the unauthorized party attempts to identify an unknown message.

The main contributions of this work consist of deriving a single-letter formula for the CR capacity for the standard two-source model with one-way communication over SISO and MIMO Gaussian channels, as well as using the obtained results on CR capacity to provide an achievable rate for correlation-assisted secure identification over Gaussian wiretap channels.

The paper is organized as follows: In Section \ref{sec2}, we provide the definition of an achievable CR rate for a model including two correlated sources with one-way communication over a SISO and MIMO Gaussian channel, respectively. Additionally, we introduce the main definitions of CR-assisted identification and secure identification. In Section \ref{sec3}, we propose a single-letter characterization of the CR capacity for the  SISO Gaussian case. We use this result to completely solve the Gaussian MIMO case by establishing the corresponding CR capacity. In Section \ref{sec4}, we derive a lower bound on the secure CR-assisted identification capacity of GWCs. Section \ref{conclusion} encompasses concluding remarks and proposes potential future research directions in this field. Auxiliary proofs are collected in the appendix.

\section{Preliminaries}
\label{sec2}
In this section, we introduce the different scenarios and channel models investigated for CR generation. Additionally, we provide some basic definitions regarding CR-assisted identification over Gaussian channels and establish the notation that will be used.
\subsection{Notation}
$\mathbb{C}$ denotes the set of complex numbers; $H(\cdot)$ and $I(\cdot ;\cdot)$ are the entropy and mutual information, respectively; $h(\cdot)$ denotes the differential entropy; all information quantities are taken to base 2; $\lVert\boldsymbol{a}\rVert_{2}$ denotes the L$_2$ norm of a vector $\boldsymbol{a}$; $\bs{A}^\h$ stands for the Hermitian transpose of the matrix $\bs{A}$, $|\mathcal{A}|$ stands for the cardinality of the set $\mathcal{A}$, $\log$ is taken to base 2 and $\ln$ stands for the natural logarithm. $\mathcal{T}_{P_{X}}^{n}$ denotes the set of typical sequences of length $n$ and of type $P_{X}$ and $\mathcal{T}_{{P}_{Y|X}}^{n}(x^n)$ denotes the set of sequences $y^n$ of length $n$ having conditional type $P_{Y|X}$ given the sequence $x^n$ of length $n$.
%and $\mathcal{N}(\mu,\sigma^2)$ refers to the Gaussian distribution with mean $\mu$ and variance $\sigma^2$.
\subsection{Common Randomness Generation: Two Correlated Sources with One-Way Communication over a Gaussian Channel}
A discrete memoryless multiple source DMMS $P_{XY}\in \mathcal{P} \left(X \times Y\right)$ with two components, with  generic variables $X$ and $Y$ on alphabets $\mathcal{X}$ and $\mathcal{Y}$, correspondingly, is given. The $n$-lengths source outputs are observable at Terminals $A$ and $B$, respectively.

Terminal $A$ generates a random variable $K=\Phi(X^n)$ with alphabet $\mathcal{K}$ and a random sequence $T^n=\Lambda(X^n).$ 
 $T^n$ is sent over a Gaussian channel with input constraint. Let $Z^n$ be the channel output. Terminal $B$ generates a random variable $L$ with the same alphabet $\mathcal{K}$  as a function of $Y^n$
and $Z^n$, i.e., $L=\Psi(Y^n,Z^n)$. Here, $\Phi, \Lambda, 
 \Psi$ refer to functions/signal processing algorithms. 
 
A pair of random variables $(K,L)$ is permissible if $K$ and $L$ are functions of the resources available at Terminal $A$ and Terminal $B$, respectively i.e.,
\begin{equation}
    K=\Phi(X^{n}), \ \     L=\Psi(Y^{n},Z^{n}).
    \label{KLSISOcorrelated}
\end{equation}
\begin{remark}
In the case of a communication over a MIMO channel with channel output $\bs{Z}^n$, a pair of random variables $(K,L)$ is permissible if $K$ and $L$ are functions of the resources available at Terminal $A$ and Terminal $B$, respectively, i.e.,
\begin{equation}
    K=\Phi(X^{n}), \ \     L=\Psi(Y^{n},\bs{Z}^{n}).
    \label{KLMIMOcorrelated}
\end{equation}
\end{remark}
\begin{definition} 
A number $H$ is called an achievable CR rate if for sufficiently large $n$ and every $\alpha >0$ $\delta>0$, there exists a permissible pair of random variables $(K,L)$ such that 
\begin{equation}
    \mbb P[K\neq L]\leq \alpha
    \label{errorcorrelated}
\end{equation}
\begin{equation}
    \frac{1}{n}H(K)> H-\delta.
     \label{ratecorrelated}
\end{equation}
\end{definition}
\begin{definition} 
The CR capacity is the maximum achievable CR rate \cite{part2}.
\end{definition}
\subsection{Gaussian Channel Model}
\label{SISOchannel}
\begin{figure}[!htb]
\centering
\tikzstyle{block} = [draw, rectangle, rounded corners,
minimum height=3em, minimum width=3cm]
\tikzstyle{blockchannel} = [draw, top color=white, bottom color=white!80!gray, rectangle, rounded corners,
minimum height=1cm, minimum width=.3cm]
\tikzstyle{input} = [coordinate]
\usetikzlibrary{arrows}
\begin{tikzpicture}[scale= 1,font=\footnotesize]
\node[blockchannel] (source) {$P_{XY}$};
\node[blockchannel, below=4cm of source](channel) { Gaussian channel};
\node[block, below left=2.2cm of source] (x) {Terminal $A$};
\node[block, below right=2cm of source] (y) {Terminal $B$};
%\node[block,right= .7cm of channel.390] (bob) {\small Decoder};
%\node[block,right=.7cm of channel.330] (eve) {\small Eavesdropper};
\node[above=1cm of x] (k) {$K=\Phi(X^n)$};
\node[above=1cm of y] (l) {$L=\Psi(Y^n,Z^n)$};

\draw[->,thick] (source) -- node[above] {$X^n$} (x);
\draw[->, thick] (source) -- node[above] {$Y^n$} (y);
\draw [->, thick] (x) |- node[below right] {$T^n=\Lambda(X^n)$} (channel);
\draw[<-, thick] (y) |- node[below left] {$Z^n$} (channel);
\draw[->] (x) -- (k);
\draw[->] (y) -- (l);

\end{tikzpicture}
\caption{ Standard Two-source model with one-way Communication over a SISO Gaussian channel}
\label{correlatedSISO}
\end{figure}
We consider first the SISO Gaussian channel, as depicted in Fig. \ref{correlatedSISO}.
Terminal $A$ encodes  $X^{n}$ into a sequence $T^{n}$ satisfying
\begin{equation}
    \mathbb{E}[\lvert T_{i} \rvert^{2}]\leq P.
    \label{energyconstraintSISOCorrelated}
\end{equation}
It follows from \eqref{energyconstraintSISOCorrelated} that each input sequence $t^n$ lies in the new constrained input set ${\mathcal{T}}_{n,P}$ defined as follows:
 \begin{multline}
\mathcal{T}_{n,P}=\{t^n \in \mathcal{T}^n \subset \mathbb{C}^n \ \text{realization of} \ T^n \\
\text{that is subject to} \ \mathbb{E}[\lvert T_{i}\rvert^{2}]\leq P \hspace{0.2cm} i=1,\ldots,n 
    \}.
    \label{inputConstrained}
\end{multline}
The sequence $T^{n}$ is sent over a Gaussian channel with an input constraint as in \eqref{energyconstraintSISOCorrelated}, and  $Z^{n}$ is defined as the  channel output, where it holds that
\begin{align}
    Z_{i}=T_{i}+ \xi_{i}  \hspace{0.3cm} i=1\ldots n,
\nonumber \end{align}
where $\xi_{i} \sim \mathcal{N}_{\mathbb{C}}(0,\sigma^2)$.
For simplicity, we drop the index $i$. The channel has capacity 
\begin{equation}
    C(P)=\underset{{\substack{P_{T}\\ \mathbb{E}[\lvert T \rvert^{2}]\leq P}}}{\max} I(T;Z).
    \label{transmissioncapacitySISOcorrelated}
\end{equation}
\subsection{MIMO Gaussian Channel Model}
\label{MIMOchannel}
We consider second the MIMO Gaussian channel, as depicted in Fig. \ref{correlatedMIMO}.
\begin{figure}[!t]
\centering
\tikzstyle{block} = [draw, rectangle, rounded corners,
minimum height=2em, minimum width=2cm]
\tikzstyle{blocksource} = [draw, top color=white, bottom color=white!80!gray, rectangle, rounded corners,
minimum height=1.1cm, minimum width=.31cm]
\tikzstyle{blockchannel} = [draw, top color=white, bottom color=white!80!gray, rectangle, rounded corners,
minimum height=1.5cm, minimum width=.35cm]
\tikzstyle{input} = [coordinate]
\tikzstyle{vectorarrow} = [thick, decoration={markings,mark=at position
   1 with {\arrow[semithick]{open triangle 60}}},
   double distance=1.4pt, shorten >= 5.5pt,
   preaction = {decorate},
   postaction = {draw,line width=1.4pt, white,shorten >= 4.5pt}]
   \tikzstyle{sum} = [draw, circle,inner sep=0pt, minimum size=2mm,  thick]
\usetikzlibrary{arrows}
\scalebox{.9}{
\begin{tikzpicture}[scale= 1,font=\footnotesize]
\node[blocksource] (source) {$P_{XY}$};
\node[blockchannel, below=5cm of source](channel) {\large$\mathbf{H}$};
\node[sum, right=.5cm of channel] (sum) {$+$};

\node[block, below left=3.15cm of source] (x) {Terminal $A$};
\node[block, below right=3.15cm of source] (y) {Terminal $B$};
%\node[block,right= .7cm of channel.390] (bob) {\small Decoder};
%\node[block,right=.7cm of channel.330] (eve) {\small Eavesdropper};
\node[above=.5cm of sum] (noise) {$\bs{\xi^n}$};
\node[above=1.5cm of x] (k) {$K=\Phi(X^n)$};
\node[above=1.5cm of y] (l) {$L=\Psi(Y^n,\bs{Z}^n)$};

\draw[->,thick] (source) -- node[above] {$X^n$} (x);
\draw[->, thick] (source) -- node[above] {$Y^n$} (y);
\draw [vectorarrow] (x) |- node[below right] {$\bs{T}^n=\Lambda(X^n)$} (channel);
\draw [vectorarrow] (channel) -- (sum);
\draw[vectorarrow] (noise) -- (sum);
\draw[vectorarrow] (sum) -| node[below left] {$\bs{Z}^n$} (y);
\draw[->] (x) -- (k);
\draw[->] (y) -- (l);

\end{tikzpicture}}
\caption{Standard Two-source model with one-way Communication over a MIMO Gaussian channel.}
\label{correlatedMIMO}
\end{figure}
Terminal $A$ encodes  $X^{n}$ into a sequence $\bs{T}^{n} \in \mathbb{C}^{N_{T}\times n},$ such that 
\begin{equation}
    \mathbb{E}[\lVert \bs{T}_{i}\rVert_{2}^{2}]\leq P,  \ \ i=1\dots n.
    \label{energyconstraintMIMOCorrelated}
\end{equation}
It follows that each input sequence $\bs{t}^n$ lies in the input set $\mathcal{T}_{N_T\times n,P},$ defined as follows:
%\begin{equation}
%\mathcal{T}_{N_T\times n,P}=\left\{\bs{t}^n \in \mathbb{C}^{N_T \times n}\ \colon \mathbb{E}[\lVert \bs{T}_{i}\rVert_{2}^{2}]\leq P \hspace{0.2cm} i=1,\ldots,n\right\}. 
%\end{equation}
\begin{multline}
\mathcal{T}_{N_T\times n,P}=\{\bs{t}^n \in \mathbb{C}^{N_T \times n}\ \text{realization of} \ \bs{T}^n\\
\text{such that} \ \mathbb{E}[\lVert \bs{T}_{i}\rVert_{2}^{2}]\leq P \hspace{0.2cm} i=1,\ldots,n 
    \}.
\end{multline}
We consider the following channel model with $N_T$ transmit antennas and  $N_R$ receive antennas:
	\begin{equation} \bs{Z}_i=\mathbf{H}\bs{T}_i+\bs{\xi}_i,\quad \forall i=1,\ldots,n,
	\label{Model:MIMOcorrelated}\end{equation}
where $n$, as previously mentioned, is the number of channel uses, as shown in Fig. \ref{correlatedMIMO}. For simplicity, we drop the index $i$.	The input vector $\boldsymbol{T} \in \Cn $ contains the $N_T$ scalar transmitted signals and fulfills the following power constraint:
\begin{equation}
     \mathbb{E}[\lVert \bs{T}\rVert_{2}^{2}]\leq P.
\nonumber \end{equation}
The output vector $\bs{Z} \in \Cr$
comprises the scalar received signals of the $N_R$ channel outputs. The channel matrix
	\begin{equation*} \mathbf{H}= \begin{pmatrix} h_{11} & \ldots& h_{1N_T} \\ 
	\vdots & \ddots & \vdots \\  h_{N_R1} & \ldots & h_{N_RN_T}
	\end{pmatrix}\in \mathbb{C}^{N_R \times N_T} \end{equation*}
	is a full-rank deterministic matrix.
	%including all effects from the input of the front-end at the transmitter to the output of the front-end at the receiver. 
	The entry $h_{ij}$ represents the channel gain from transmit antenna $j$ to receive antenna $i$. 
The vector $\bs{\xi} \in \Cr$ is the circularly symmetric Gaussian noise, $\bs{\xi} \sim \mathcal{N}_{\mathbb{C}}(\bs{0}_{N_R},\sigma^2 \mathbf{I}_{N_R})$.

The MIMO channel has the capacity 
\begin{equation}
   C(P,N_{T} \times N_{R})=\underset{{\substack{P_{\bs{T}}\\ \mathbb{E}[\lVert \bs{T}\rVert_{2}^{2}]\leq P}}}{\max} I(\bs{T};\bs{Z}).
    \label{transmissioncapacityMIMOcorrelated}
\end{equation}

\subsection{CR-assisted Identification over Gaussian Channels}
In 1989, Ahlswede and Dueck \cite{Idchannels} proposed  the identification scheme which is conceptually different from the classical transmission scheme of Shannon. In transmission, the encoder transmits a message over a channel $W,$ and at the receiver side, the decoder wants to estimate this message based on the channel observation. However, this is not the case for identification. Indeed, in the identification scheme, the encoder sends an identification message (also called an identity) $M \in \mathcal{N}$ over the channel and the decoder is not interested in \emph{what} the received message is, but he wants to check \emph{whether} a specific message $\hat{M} \in \mathcal{N}$ has been sent or not.
Naturally, the sender has no knowledge of this specific message, otherwise it would be a trivial problem. The identification problem can be regarded as solving many hypothesis testing problems occurring simultaneously. 

There are many interesting applications of the identification scheme, such as in industry 4.0, online sales, and the healthcare field \cite{application}. For instance, in product engineering, sensors are used to control the sequence of production. The sensor data is collected and processed by a central unit. Here, the receiver is interested in checking whether or not an error occurs in the sequence of production rather than determining the accurate sensor measurements.
For insight into the explicit construction of identification codes, we refer the reader to \cite{explicitconstruction}. Furthermore, a special identification code construction using tag codes with two concatenated Reed-Solomon codes is implemented in \cite{implementation}.
In our work, we are particularly interested in studying the problem of CR-assisted identification, in which the transmitter and the receiver have access to a correlated source
$P_{XY}\in \mathcal{P} \left(X \times Y\right)$ as visualized in Fig.  \ref{fig:IdwithCR}.
 \begin{figure*}[!t]
    \centering
    \tikzstyle{block} = [draw, top color=white, bottom color=white!80!gray, rectangle, rounded corners,
minimum height=2em, minimum width=3cm]
\tikzstyle{blockchannel} = [draw, top color=white, bottom color=white!80!gray, rectangle, rounded corners,
minimum height=2cm, minimum width=.1cm]
\tikzstyle{input} = [coordinate]
\usetikzlibrary{arrows}
\scalebox{.74}{\begin{tikzpicture}[auto, node distance=2cm,>=latex']
\node[] (m) {\small $M$};
\node[block,right=1cm of m] (enc) {\small Encoder};
\node[blockchannel, right=1.5cm of enc](channel) {\small Channel
$W^n$};
\node[block, above=.7cm of channel] (source) {$P_{XY}$};
\node[block,right= 1.5 cm of channel.360] (bob) {\small Decoder};
\node[right=1cm of bob] (what) {\small Is $\hat{M}$ sent? Yes,No?};

\draw[->] (m) -- (enc);
\draw[->] (enc) -- node[above]{$T^{n}$}  (channel);
\draw[->] (channel.360) -- node[above]{$Z^{n}$} (bob);
\draw[->] (bob) -- (what);
\draw[->] (source) -- node[above]{$X^n$} (enc);
\draw[->] (source) -- node[above]{$Y^n$} (bob);
\end{tikzpicture}}
    \caption{CR-assisted identification.}
    \label{fig:IdwithCR}
\end{figure*}
Unlike in \cite{Idchannels}, we do not assume the existence of local randomness.
In what follows, $x^n$ and $y^n$ are realizations of $X^n$ and $Y^n$, respectively.
\begin{definition} A CR-assisted  $(n,N,\lambda_1,\lambda_2)$ identification code  for the Gaussian channel $W$ is a family of pairs $\left\{\left(\bs{u}_i,\setd_i(y^n)\right),\quad i=1,\ldots,N \right\},$ with 
\begin{equation*} \bs{u}_i=\Phi\left(x^n\right) \in \mathcal{T}_{n,P} ,\quad \mathcal D_{i}(y^n) \subset {\mathcal{Z}} ^{n},~\forall ~i\in \{1,\ldots,N\}, \end{equation*}
such that for all $i,j \in \{1,\ldots,N\}$, $i \neq j$ and $\lambda_1+\lambda_2 <1,$ the errors of the first and second kind  satisfy
\begin{align}
&  W^n((\setd_i(y^n))^c|\bs{u}_i)  \leq\lambda_1, \label{errofirstkind} \\
&  W^n(\setd_i(y^n)|\bs{u}_j) \leq\lambda_2 \quad \forall i\neq j. \label{errorsecondkind}
\end{align} \end{definition}
%The rate $R > 0$ of such a CR-assisted $(n,;N)$-ID code is
%said to be achievable if for all $\tau > 0$ there exist an $n(\tau) \in \mathbb{N}$
%and a sequence of $(n,;N)$-ID codes  such that for all $n\geq n(\tau)$
%we have $\frac{1}{N} \log \log(N)> R$ and the errors of first and second
%kind \eqref{errofirstkind} and \eqref{errorsecondkind} are satisfied.

%Let $N(n,\lambda)$ be the maximal number such that an $(n,N,\lambda_1,\lambda_2)$ CR-assisted identification code for the channel $W$ exists with $\lambda_1, \lambda_2 \leq \lambda$ and $\lambda \in (0,\frac{1}{2})$ then the CR-assisted identification capacity for the Gaussian Channel $W$ is given by:
%\begin{equation}
%C_{ID}(P) \triangleq \lim_{n \to \infty} \frac{1}{n} \log \log N (n, \lambda)
%\end{equation}
\begin{definition}
    \label{defcapacity}
  $C_{ID}^{c}(P)$  the  CR-assisted  identification capacity of the channel $W$ is defined as follows:
  \begin{align}
      &C_{ID}^{c}(P)= \nonumber \\
      &\max\left\{ R\colon \forall \lambda>0,\ \exists n(\lambda) \text{ s.t. for } n \geq n(\lambda)  \ N(n,\lambda) \geq 2^{2^{nR}}         \right\}, 
  \nonumber \end{align}
 where $N(n,\lambda)$ is the maximal cardinality such that a $(n,N,\lambda_1,\lambda_2)$ CR-assisted identification code for the channel $W$ exists. 
\end{definition}
 \label{RandomizedMIMO}
\begin{definition} 
For the MIMO channel described in \eqref{Model:MIMOcorrelated}, a CR-assisted $(n,N,\lambda_1,\lambda_2)$ identification-code is a family of pairs $\{\left(\mathbf{u}_i,\setd_i(y^n)\right),\ i=1,\dots,N\},$ such that for some $\lambda_1+\lambda_2 < 1$ and for all $i \in \{1,\ldots,N\}$, we have
\begin{align}
& \mathbf{u}_i \in \mathcal{T}_{N_T\times n,P}, \nonumber \\
& \setd_i(y^n) \subset \mathcal{Z}^n=\left\{\bs{z}^n=(\bs{z}_1,\bs{z}_2,\ldots,\bs{z}_n) \in \mathbb{C}^{N_R\times n}   \right\}, \nonumber
\end{align}
and with errors of the first and second kind that satisfy
\begin{align}
&  W^n((\setd_i(y^n))^c|\mathbf{u}_i)  \leq \lambda_1 \nonumber \\
&  W^n(\setd_i(y^n)|\mathbf{u}_j) \leq \lambda_2, \forall i \neq j.
\nonumber \end{align}\end{definition}
\begin{remark}
The definitions of identification codes for the single-user MIMO channel are similar to the SISO case, except for the dimension of input and output sets. Indeed, at each time instant $i \in \{1,\ldots,n\}$, we send $N_T$ scalar signals and receive $N_R$ signals. Thus, compared to the SISO case, the input and output sets contain matrices instead of vectors.
\end{remark}
 \subsection{CR-assisted Secure Identification over Gaussian Wiretap Channels}
 \begin{figure*}[!t]
    \centering
    \tikzstyle{block} = [draw, top color=white, bottom color=white!80!gray, rectangle, rounded corners,
minimum height=2em, minimum width=3cm]
\tikzstyle{blockchannel} = [draw, top color=white, bottom color=white!80!gray, rectangle, rounded corners,
minimum height=2cm, minimum width=.1cm]
\tikzstyle{input} = [coordinate]

\usetikzlibrary{arrows}
\scalebox{.68}{
\begin{tikzpicture}[auto, node distance=2cm,>=latex']
\node[] (m) {\small $M$};
\node[block,right=1cm of m] (enc) {\small Encoder};
\node[blockchannel, right=1.5cm of enc](channel) {\small Channel
$(W_{g},V{g'})$};
\node[block, above=.7cm of channel] (source) {$P_{XY}$};
\node[block,right= 1.5 cm of channel.390] (bob) {\small Decoder};
\node[block,right=1.5 cm of channel.330] (eve) {\small Eavesdropper};
\node[right=1cm of bob] (what) {\small Is $\hat{M}$ sent? Yes,No?};
\node[right=1cm of eve] (wbar) {\small Is $M'$ sent? Yes,No? }; %we can use \cancel{W} but it doesn't work with package cancel
\node[coordinate,right=5cm of source] (i) {};
\node[coordinate, below=2cm of i] (g) {};
\draw[->] (m) -- (enc);
\draw[->] (enc) -- node[above]{$T^{m}$}  (channel);
\draw[->] (channel.390) -- node[above]{$Z^{m}$} (bob);
\draw[->] (channel.330) -- node[above]{$Z'^{m}$} (eve);
\draw[->] (bob) -- (what);
\draw[->] (eve) -- (wbar);
\draw[->] (source) -- node[above]{$X^n$} (enc);
\draw[->] (source) -- node[above]{$Y^n$} (bob);
\draw[-] (source) -- node[above right]{$(X^n,Y^n)$}(i);
\draw[-, dashed] (i) -- (g);
\draw[->] (g) -|  (eve);
\end{tikzpicture}}
    \caption{CR-assisted secure identification.}
    \label{fig:secureId}
\end{figure*}
We focus now on CR-assisted secure identification
depicted in Fig. \ref{fig:secureId}. We consider the following standard model of the Gaussian wiretap channel (GWC):
\begin{equation}
\begin{aligned}\Wg\colon Z_i =T_i+\xi_i, &&  \forall i \in \{1,\ldots,m\}, \\
\Vg \colon \Zp_i =T_i + \phi_i, && \forall i \in \{1,\ldots,m\},
 \end{aligned} \label{channelModel}
 \end{equation} where $T^m=(T_1,T_2,\ldots,T_m)$ corresponds to the channel input sequence, and where  $Z^m=(Z_1,Z_2,\ldots,Z_m)$ and $Z^{\prime m}=(\Zp_1,\Zp_2,\ldots,\Zp_m)$ are Bob and Eve's observations, respectively. $\xi^m=(\xi_1,\xi_2,\ldots,\xi_m)$ and $\phi^m=(\phi_1,\phi_2,\ldots,\phi_m)$ are the noise sequences of the main channel and the wiretapper's channel, respectively. $m$ denotes the number of channel uses.
 The $\xi_i's$ are i.i.d. and each $\xi_i$ is drawn from a normal distribution denoted by $\text{g}$ with zero-mean and variance $\sigma^2$. The $\phi_i's$ are i.i.d. and each $\phi_i$ is drawn from a normal distribution denoted by $\text{g}'$ with zero-mean and variance $\sigma'^2$. The channel input fulfills the following power constraint: \begin{equation}  \mathbb{E}[|T_{i}|^{2}]\leq P  \ \ i=1\dots m. \label{powerConstr}
 \end{equation}The input set is $\mathcal{T}_{m,P}$, defined in \eqref{inputConstrained}. 
The output sets are infinite $\setz=\setz^\prime= \mathbb{C}$. We denote the GWC by the pair $(W_{\text{g}},V_{\text{g}'})$, where $\Wg$ and $\Vg$ define the Gaussian channels to the legitimate receiver and the wiretapper, with capacities $C(\text{g},P)$ and $C(\text{g}',P),$ respectively.

\begin{definition} \label{IDwiretapCode}
 A CR-assisted $(m,N,\lambda_1,\lambda_2)$ identification code for the GWC $(\Wg,\Vg)$ is a family of pairs $\{(\bs{u}_i,\setd_i(y^n)),\ i=1,\dots,N\},$ with \begin{align*} \bs{u}_i=\Phi(x^n) \in \mathcal{T}_{m,P},\ \mathcal D_{i}(y^n) \subset {\setz} ^{m}, && \forall ~i\in \{1,\ldots,N\}, \end{align*}
such that for all $i,j \in \{1,\ldots,N\}$, $i \neq j$ and some $\sete=\sete(x^n,y^n) \in {\setz^{\prime}}^m$and $\lambda_{1},\lambda_{2}\leq \frac{1}{2},$ the errors of first and second kind satisfy, respectively:
\begin{align}
\Wg^m((\setd_i(y^n))^c|\bs{u}_i)  &  \leq\lambda_1,  \label{firstRequirement} \\
\Wg^m(\setd_i(y^n)|\bs{u}_j)  &\leq\lambda_2, \label{secondRequirement}
 \end{align}
 \text{and for} $\lambda \leq \frac{1}{2},$ \text{it holds that}
 \begin{align}
 \Vg^m(\sete|\bs{u}_j)   + \Vg^m(\sete^c|\bs{u}_i) & \geq 1-\lambda. \label{wiretapCond}
\end{align}
\end{definition}
\begin{remark}
The last line \eqref{wiretapCond} means that the wiretapper cannot identify the identification message $i$ \cite{Igor}.
\end{remark}
%\begin{remark}
%In case the eavesdropper knows the correlation source signals, then the definition of CR-assisted identification code for the GWC is analogous  to Definition  \ref{IDwiretapCode}. The only difference consists in the set $\sete$, which now depends on the sequences $x^n$ and $y^n$ such that $ \sete(x^n,y^n)\in \mathcal{Z}'.$
%\end{remark}
\begin{definition}
  $C_{SID}^{c}(g,g',P)$, the secure CR-assisted  identification capacity of the channel $(W_{g},V_{g'}),$ is defined as follows:
  \begin{align}
  &C_{SID}^{c}(g,g',P)= \nonumber \\
 &\max\left\{ R\colon \forall \lambda>0,\ \exists n(\lambda) \text{ s.t. for } n \geq n(\lambda) \ N_S(n,\lambda) \geq 2^{2^{nR}}         \right\}, \nonumber
\end{align}
where $N_S(n,\lambda)$ is the maximal cardinality such that a $(n,N,\lambda_1,\lambda_2)$ CR-assisted identification wiretap code for the channel $(W_{g},V_{g'})$ exists. 
\end{definition}
\begin{remark}
\label{remarkcorrelation}
As correlation cannot increase the Shannon message-transmission capacity, it is not utilized in current communication systems. However, this is not the case for identification. We will demonstrate in Section \ref{sec4} that for the identification task, we can achieve performance gains by taking advantage of CR.
\end{remark}
\section{Common Randomness Capacity}
\label{sec3}
In this section, we propose a single-letter characterization of the CR capacity for the scenarios presented in the previous section and provide a rigorous proof of it.
\subsection{SISO Case}
We start with the first scenario depicted in Fig. \ref{correlatedSISO}, where the communication is over a SISO Gaussian channel with the power constraint defined in $\eqref{energyconstraintSISOCorrelated}$.
\begin{proposition}
\text{For the model in Fig}. \ref{correlatedSISO}, the CR capacity $C_{CR}(P)$ is equal to
\begin{align}
C_{CR}(P)= \underset{\substack{U \\{\substack{U \circlearrow{X} \circlearrow{Y}\\ I(U;X)-I(U;Y) \leq C(P)}}}}{\max} I(U;X).  \label{CRcapacity formula}
\end{align}
\label{proposition1}
\end{proposition}
\subsubsection{Direct Proof:} We extend the coding scheme provided in \cite{part2} to Gaussian channels. By continuity, it suffices to show that 
$$ \underset{ \substack{U \\{\substack{U \circlearrow{X} \circlearrow{Y}\\ I(U;X)-I(U;Y) \leq R'}}}}{\max} I(U;X)  $$ is an achievable CR rate for every $R'<C(P).$
Let $U$ be a random variable satisfying $U \circlearrow{X} \circlearrow{Y}$ and $I(U;X)-I(U;Y) \leq R'.$ We are going to show that $H=I(U;X)$ is an achievable CR rate. Let $\alpha,\delta>0$. Without loss of generality, assume that the distribution of $U$ is a possible type for block-length $n$.
For some $\mu>0,$ we let
\begin{align}
N_{1}&=\lfloor 2^{n[I(U;X)-I(U;Y)+3\mu]} \rfloor \nonumber
\end{align}
and
\begin{align}
N_{2}&=\lfloor 2^{n[I(U;Y)-2\mu]}\rfloor. \nonumber
\end{align}For each pair $(i,j)$ with $1\leq i \leq N_{1}$ and $1\leq j \leq N_{2}$, we define a random sequence $\bs{U}_{i,j}\in\mathcal{U}^n$ of type $P_{U}$. Let $\mbf M=\bs{U}_{1,1},\hdots, \bs{U}_{N_{1},N_{2}}$  be the joint random variable of all $\bs{U}_{i,j}s.$ We define $\Phi_{\mbf M}$ as follows:
 Let $\Phi_{\mbf M}(X^n)=\bs{U}_{ij}$, if $\bs{U}_{ij}$ is jointly $UX$-typical with $X^n$ (either one if there are several). If no such $\bs{U}_{i,j}$ exists, then  $\Phi_{\mbf M}(X^n)$ is set to a constant sequence $\bs{u}_0$ different from all the  ${\bs{U}_{ij}}s$, jointly $UX$-typical with none of the realizations of $X^n$ and known to both terminals.
 
We further define the following two sets which depend on $\mbf M$:
\begin{align}
    S_{1}(\mbf M)&=\{(x^n,y^n):(\Phi_{\mbf M}(x^n),x^n,y^n) \in \mathcal{T}_{U,X,Y}^{n}\} \nonumber
\end{align} and
\begin{align}
    S_{2}(\mbf M) 
    &=\Big\{(x^n,y^n):(x^n,y^n)  \in S_{1}(\mbf M) \ \text{s.t.} \ \bs{U}_{i,j}=\Phi_{\mbf M}(x^n) \nonumber \\   & \ \ \ \  \ \text{and} \ \exists \ \bs{U}_{i,\ell}\neq\bs{U}_{i,j} \ \text{jointly} \ UY\text{-typical with} \ y^n \nonumber \\
    & \ \ \ \ \ (\text{with the same first index} \ i)
\Big\}.\nonumber
\end{align}
It is proved in \cite{part2} that 
\begin{align}
    \mathbb{E}_{\mbf M}\left[ \mbb P\left[(X^n,Y^n)\notin  S_{1}(\mbf M)\right]+\mbb P\left[(X^n,Y^n)\in  S_{2}(\mbf M)\right]\right] \leq \beta(n),
    \label{averagebeta}
\end{align}
where $\beta(n) \leq \frac{\alpha}{4}$ for sufficiently large $n$. 
We choose a realization $$\mbf m=\bs{u}_{1,1},\hdots, \bs{u}_{N_1,N_2}$$ satisfying
\begin{align}
\mbb P\left[(X^n,Y^n)\notin  S_{1}(\mbf m)\right]+\mbb P\left[(X^n,Y^n)\in  S_{2}(\mbf m)\right]\leq 2\beta(n). \label{choiceofm}
\end{align} 
From \eqref{averagebeta} and using Markov inequality, we know that such a realization exists. We denote $\Phi_{\mbf m}$ by $\Phi.$
We assume that each $\bs{u}_{i,j}, i=1\hdots N_1, j=1\hdots N_2,$  is known to both terminals.  
This means that  $N_{1}$ codebooks $C_{i}, 1\leq i \leq N_{1}$, are known to both terminals, where each codebook contains $N_{2}$ sequences, $ \bs{u}_{i,j}, \ j=1,\hdots, N_2$. 

Let $x^n$ be any realization of $X^n$ and $y^n$ be any realization of $Y^n.$
 %Let $\Phi(x^n)=\bs{u}_{ij}$, if $\bs{u}_{ij}$ is jointly $UX$-typical with $x^n$ (either one if there are several).
 Let $f_1(x^n)=i$ if $\Phi(x^n)=\bs{u}_{i,j}$. Otherwise, if $\Phi(x^n)=\bs{u}_{0},$ then $f_1(x^n)=N_1+1.$
 %different from all the ${\bs{u}_{ij}}s$, jointly $UX$-typical with none of the realizations of $X^n$ and known to both terminals.
  Since $ R'<C(P)$, we choose $\mu$ to be sufficiently small such that
      \begin{align}
     \frac{\log \lVert f_1 \rVert}{n}&=\frac{\log(N_1+1)}{n} \nonumber \\
     &\leq C(P)-\mu',
     \label{inequalitylogfSISO}
      \end{align}
for some $\mu'>0,$
 The message $i^\star=f_1(x^n)$, with $i^\star\in\{1,\hdots,N_1+1\},$ is encoded to a sequence $t^n$ using a code sequence $(\Gamma^\star_n)_{n=1}^{\infty},$ using a suitable \textit{forward error correcting code}, with rate $\frac{\log \lVert \Gamma^\star_n \rVert}{n}=\frac{\log \lVert f_1 \rVert}{n}$ satisfying \eqref{inequalitylogfSISO}
 and with maximum error probability not exceeding $\frac{\alpha}{2}$ for sufficiently large $n.$
  Here, $\lVert f_1 \rVert$ refers to the cardinality of the set of messages $\{i^\star:i^\star=1,\hdots,N_1+1\}
$.
   The sequence $t^n$ is sent over the Gaussian channel. Let $z^n$ be the corresponding channel output sequence. Terminal $B$ decodes the message $\tilde{i}^\star$ from the knowledge of $ z^n.$
Let $\Psi(y^n,z^n)=\bs{u}_{\tilde{i}^\star,j}$ if $\bs{u}_{\tilde{i}^\star,j}$ and $y^n$ are jointly $UY$-typical . If there is no such  $\bs{u}_{\tilde{i}^\star,j}$ or there are several, we set $\Psi(y^n,z^n)=\bs{u}_0$ (since $K$ and $L$ must have the same alphabet).
Now, we are going to show that the requirements in $\eqref{errorcorrelated}$ and $\eqref{ratecorrelated}$ are satisfied.
We define next for any $(i,j)\in \{1,\hdots,N_1\}\times\{1,\hdots,N_2\}$  the set
$$\mc S=\{ x^n\in\mathcal{X}^{n} \ \text{s.t.} \ (\bs{u}_{i,j},x^n) \ \text{jointly} \ UX\text{-typical}\}.$$
Then, it holds that 
\begin{align}
&\mbb P[K=\bs{u}_{i,j}] \nonumber \\ &=\sum_{x^n\in\mc S}\mbb P[K=\bs{u}_{i,j}|X^n=x^n]P_{X}^n(x^n) \nonumber \\  &\quad+\sum_{x^n\in\mc S^c}\mbb P[K=\bs{u}_{i,j}|X^n=x^n]P_{X}^n(x^n) \nonumber \\
&\overset{(a)}{=}\sum_{x^n\in\mc S}\mbb P[K=\bs{u}_{i,j}|X^n=x^n]P_{X}^n(x^n) \nonumber \\
&\leq \sum_{x^n\in\mc S}P_{X}^n(x^n) \nonumber \\
&=P_{X}^{n}(\{x^n: (\bs{u}_{i,j},x^n) \ \text{jointly} \ UX\text{-typical}\}) \nonumber\\
& = 2^{-nI(U;X)-\kappa(n)}, \nonumber
\end{align}
for some $\kappa(n)>0$ with $\underset{n\rightarrow \infty}{\lim} \frac{\kappa(n)}{n}=0$,
where $(a)$ follows because for  $(\bs{u}_{i,j},\mathbf{x})$ being not jointly $UX$-typical, we have $\mbb P[K=\bs{u}_{i,j}|X^n=x^n]=0.$ This yields
\begin{align}
H(K)\geq nI(U;X)-\kappa'(n)
\nonumber \end{align}
for some $\kappa'(n)>0$ with $\underset{n\rightarrow \infty}{\lim} \frac{\kappa'(n)}{n}=0.$
Therefore, for sufficiently large $n,$ it holds that
\begin{align}
    \frac{H(K)}{n}>H-\delta. \nonumber
\end{align}
Thus, (\ref{ratecorrelated}) is satisfied. Now, it remains to prove that \eqref{errorcorrelated} is satisfied. For this purpose, we define the following event:
\begin{align}
    \mathcal{D}_{\mbf m}= ``\Phi(X^n) \ \text{is equal to none of the} \  {\bs{u}_{i,j}}s". \nonumber
\end{align}
We denote its complement by $\mc D_{\mbf m}^{c}.$
We further define $I^\star=f_1(X^n)$ to be the random message generated by Terminal $A$ and  $\tilde{I}^\star$ to be the random message decoded by Terminal $B$. 
We have
\begin{align}
    &\mbb P[K\neq L] \nonumber \\ &=\mbb P[K\neq L|I^\star=\tilde{I}^\star]\mbb P[I^\star=\tilde{I}^\star] \nonumber \\ &\quad+ \mbb P[K\neq L|I^\star\neq \tilde{I}^\star]\mbb P[I^\star\neq\tilde{I}^\star] \nonumber \\
        &\leq \mbb P[K\neq L|I^\star=\tilde{I}^\star]+ \mbb P[I^\star\neq\tilde{I}^\star].\nonumber
\end{align}
Here,
\begin{align}
    &\mbb P[K\neq L|I^\star=\tilde{I}^\star] \nonumber \\
   &= \mbb P[K\neq L|I^\star=\tilde{I}^\star,\mathcal{D}_{\mbf m}]\mbb P[\mathcal{D}_{\mbf m}|I^\star=\tilde{I}^\star] \nonumber \\ &\quad + \mbb P[K\neq L|I^\star=\tilde{I}^\star,\mathcal{D}_{\mbf m}^c]\mbb P[\mathcal{D}_{\mbf m}^c|I^\star=\tilde{I}^\star] \nonumber \\
   &\overset{(a)}{=}\mbb P[K\neq L|I^\star=\tilde{I}^\star,\mathcal{D}_{\mbf m}^c]\mbb P[\mathcal{D}_{\mbf m}^c|I^\star=\tilde{I}^\star] \nonumber \\
   &\leq \mbb P[K\neq L|I^\star=\tilde{I}^\star,\mathcal{D}_{\mbf m}^c],\nonumber
\end{align}
where $(a)$ follows from $\mbb P[K\neq L|I^\star=\tilde{I}^\star,\mathcal{D}_{\mbf m}]=0,$ since conditioned on $I^\star=\tilde{I}^\star$ and $\mathcal{D}_{\mbf m}$, we know that $K$ and $L$ are both equal to $\bs{u}_0.$ Thus, we obtain
\begin{align}
    \mbb P[K\neq L] 
    &\leq \mbb P[K\neq L|I^\star=\tilde{I}^\star,\mathcal{D}_{\mbf m}^c]+ \mbb P[I^\star\neq\tilde{I}^\star] \nonumber \\
    &\leq \mbb P\left[ \mc A(\mbf m)\cup \mc B(\mbf m)\right]+\mbb P[I^\star\neq\tilde{I}^\star] \nonumber \\
    &\overset{(a)}{=}\mbb P\left[\mc A(\mbf m)\right]+\mbb P\left[\mc B(\mbf m)\right] +\mbb P[I^\star\neq\tilde{I}^\star] \nonumber \\
    &\overset{(b)}{\leq} 2\beta(n)+ \mbb P[I^\star\neq\tilde{I}^\star] \nonumber \\
    &\leq \alpha,
\nonumber \end{align}
where $(a)$ follows because the events $\mc A(\mbf m)$ and $\mc B(\mbf m)$ are independent, $(b)$ follows from \eqref{choiceofm}
and $(c)$ follows because $2\beta(n)+\theta\leq \alpha$ for sufficiently large $n.$
 This proves \eqref{errorcorrelated}. 
This completes the achievability proof.
\subsubsection{Converse Proof:}
Let $H$ be any achievable CR rate. So, for every $\alpha,\delta>0$ and for sufficiently large $n,$ there exists a permissible pair of random variables $(K,L)$ according to a fixed CR-generation protocol of block-length $n$  such that 
\begin{equation}
  \mbb P\left[K\neq L\right]\leq \alpha, 
    \label{errorcorrelated1}
\end{equation}
and
\begin{equation}
    \frac{1}{n}H(K)> H-\delta.
     \label{ratecorrelated1}
\end{equation}
In our proof, we will use  the following lemma: 
\begin{lemma} (Lemma 17.12 in \cite{codingtheorems})
 	For arbitrary random variables $R_1$ and $R_2$ and sequences of random variables $X^{n}$ and $Y^{n}$, it holds that
 	\begin{align}
 	&I(R_1;X^{n}|R_2)-I(R_1;Y^{n}|R_2) \nonumber \\
 	&=\sum_{i=1}^{n} I(R_1;X_{i}|X_{1},\dots, X_{i-1}, Y_{i+1},\dots, Y_{n},R_2) \nonumber \\ &\quad -\sum_{i=1}^{n} I(R_1;Y_{i}|X_{1},\dots, X_{i-1}, Y_{i+1},\dots, Y_{n},R_2) \nonumber \\
 	&=n[I(R_1;X_{J}|V)-I(R_1;Y_{J}|V)],\nonumber
 	\end{align}
 	where $V=(X_{1},\dots, X_{J-1},Y_{J+1},\dots, Y_{n},R_2,J)$, with $J$ being a random variable independent of $R_1$,\ $R_2$, \ $X^{n}$ \ and $Y^{n}$ and uniformly distributed on $\{1 ,\dots, n \}$.
\label{lemma1}
\end{lemma}Let $J$ be a random variable uniformly distributed on $\{1,\dots, n\}$ and independent of $K$, $X^n$ and $Y^n$. We further define $U=(K,X_{1},\dots, X_{J-1},Y_{J+1},\dots, Y_{n},J).$ It holds that $U \circlearrow{X_J} \circlearrow{Y_J}.$ 
 Notice  that
 \begin{align}
 		\frac{1}{n}H(K)&\overset{(a)}{=}\frac{1}{n}H(K)-\frac{1}{n}H(K|X^{n})\nonumber\\
 		&=\frac{1}{n}I(K;X^{n}) \nonumber\\
 		&\overset{(b)}{=}\frac{1}{n}\sum_{i=1}^{n} I(K;X_{i}|X_{1},\dots, X_{i-1}) \nonumber\\
 		&=I(K;X_{J}|X_{1},\dots, X_{J-1},J) \nonumber\\
 		&\overset{(c)}{\leq }I(U;X_{J}), \nonumber
 		\end{align}
   where $(a)$ follows because $K=\Phi(X^n)$ and $(b)$ and $(c)$ follow from the chain rule for mutual information.
\color{black} 		
 Applying Lemma \ref{lemma1} for $R_1=K$, $R_2=\varnothing$ with $V=(X_1,\hdots, X_{J-1},Y_{J+1},\hdots, Y_{n},J)$ yields 
 \begin{align}
 		&\frac{1}{n}\left[I(K;X^{n})-I(K;Y^{n})\right] \nonumber \\
 		&=I(K;X_{J}|V)-I(K;Y_{J}|V) \nonumber\\
 		&\overset{(a)}{=}I(KV;X_{J})-I(V;X_{J})-I(KV;Y_{J})+I(V;Y_{J})\nonumber\\ 
 		&\overset{(b)}{=}I(U;X_{J})-I(U;Y_{J}), 
 		\label{UhilfsvariableMIMO1}
 		\end{align}where $(a)$ follows from the chain rule for mutual information and from the fact that $V$ is independent of $(X_{J},Y_{J})$ and $(b)$ follows from $U=(K,V)$.  
   It results using (\ref{UhilfsvariableMIMO1}) that
 \begin{align}
 		I(U;X_{J})-I(U;Y_{J})
 		&=\frac{1}{n}\left[I(K;X^{n})-I(K;Y^{n})\right] \nonumber\\
 		&=\frac{1}{n}H(K)-\frac{1}{n}I(K;Y^{n})\nonumber \\ 
 		&=\frac{1}{n}H(K|Y^n) 
 		\nonumber \\
            &=\frac{1}{n}H(K|Y^n,Z^n)+\frac{1}{n}I(K;Z^n|Y^n) \nonumber \\
            &\leq \frac{1}{n}H(K|L) +\frac{1}{n}I(K;Z^n|Y^n)\nonumber \\
       &\overset{(a)}{\leq} \frac{1}{n}\left(1+\log\lvert \mc K\rvert \mbb P\left[K\neq L\right]\right) +\frac{1}{n}I(K;Z^n|Y^n) \nonumber \\
       &\overset{(b)}{\leq} \frac{1}{n}+\alpha \log\lvert \mc X \rvert+\frac{1}{n}I(K;Z^n|Y^n). \nonumber
   \end{align}
   where $(a)$ follows from Fano's inequality and $(b)$ follows from \eqref{errorcorrelated1} and from the fact $\lvert \mc K \rvert \leq \lvert \mc X \rvert^{n}.$
   On the one hand, we have
\begin{align} 
\frac{1}{n}I(K;Z^{n}|Y^{n})&\leq \frac{1}{n}I(X^{n},K;Z^{n}|Y^{n}) \nonumber\\
& \overset{(a)}{\leq } \frac{1}{n}I(T^{n};Z^{n}|Y^{n})  \nonumber \\
& = \frac{1}{n}h(Z^{n}|Y^{n})- \frac{1}{n}h(Z^{n}|T^{n},Y^{n}) \nonumber \\
& \overset{(b)}{=} \frac{1}{n}h(Z^{n}|Y^{n})- \frac{1}{n}h(Z^{n}|T^{n}) \nonumber \\
& \overset{(c)}{\leq }  \frac{1}{n}h(Z^{n})- \frac{1}{n}h(Z^{n}|T^{n}) \nonumber \\
& = \frac{1}{n}I(T^{n};Z^{n}),  \nonumber \end{align}
where $(a)$ follows from the Data Processing Inequality because $Y^{n}\circlearrow{X^{n}K}\circlearrow{T^n}\circlearrow{Z^{n}}$ forms a Markov chain, $(b)$ follows because $Y^{n}\circlearrow{X^{n}K}\circlearrow{T^n}\circlearrow{Z^{n}}$ forms a Markov chain, $(c)$ follows because conditioning does not increase entropy,
On the other hand, we have:
\begin{align}
& \frac{1}{n}I(T^{n};Z^{n})  \nonumber \\
& \overset{(a)}{=} \frac{1}{n}\sum_{i=1}^{n} I(Z_{i};T^{n}|Z^{i-1}) \nonumber \\
& = \frac{1}{n}\sum_{i=1}^{n} h(Z_{i}|Z^{i-1})-h(Z_{i}|T^{n},Z^{i-1}) \nonumber \\
& \overset{(b)}{=} \frac{1}{n}\sum_{i=1}^{n} h(Z_{i}|Z^{i-1})-h(Z_{i}|T_{i}) \nonumber \\
& \overset{(c)}{\leq} \frac{1}{n}\sum_{i=1}^{n} h(Z_{i})-h(Z_{i}|T_{i}) \nonumber \\
& = \frac{1}{n}\sum_{i=1}^{n} I(T_{i};Z_{i}) \nonumber \\
& \overset{(d)}{\leq} C(P),
\nonumber \end{align}
 where $(a)$ follows from the chain rule for mutual information,
$(b)$ follows because $$T_{1},\dots, T_{i-1},T_{i+1},\dots, T_{n},Z^{i-1} \circlearrow{T_{i}}\circlearrow{Z_{i}}$$ forms a Markov chain, $(c)$ follows because conditioning does not increase entropy and
$(d)$ follows from \eqref{transmissioncapacitySISOcorrelated} and from the fact that  $T_{i}, i=1,\hdots, n$ satisfies the power constraint in \eqref{energyconstraintMIMOCorrelated}. 
   Thus, we obtain
   \begin{align}
       I(U;X_{J})-I(U;Y_{J})\leq C(P)+\zeta(n,\alpha),
   \nonumber \end{align}
where $\zeta(n,\alpha)=\frac{1}{n}+\alpha \log\lvert \mc X \rvert.$
Since the joint distribution of $X_{J}$ and $Y_{J}$ is equal to $P_{XY},$ it follows that $\frac{H(K)}{n}$ is upper-bounded by $I(U;X)$ subject to $I(U;X)-I(U;Y) \leq C(P) + \zeta(n,\alpha)$ with $U$ satisfying $U\circlearrow{X} \circlearrow{Y}.$ As a result, it follows using \eqref{ratecorrelated1} that any achievable CR rate satisfies 
\begin{align}
H <\underset{ \substack{U \\{\substack{U\circlearrow{X} \circlearrow{Y}\\ I(U;X)-I(U;Y) \leq C(P)+\zeta(n,\alpha)}}}}{\max} I(U;X)+\delta.
 \label{righthandsideconverse}
\end{align}
By taking the limit when $n$ tends to infinity and then the infimum over all $\alpha>0,\delta>0,$ of the right-hand side of \eqref{righthandsideconverse}, it follows that
\begin{align}
&H \leq \underset{ \substack{U \\{\substack{U \circlearrow{X} \circlearrow{Y}\\ I(U;X)-I(U;Y) \leq C(P)}}}}{\max} I(U;X). \nonumber
\end{align}
This completes the converse proof.

\begin{remark}
There exists  $P_{\star}$ such that $C(P_{\star})=H(X|Y)$, where 
\begin{align}
    C_{CR}(P)=C_{CR}(P_{\star})=H(X) \quad \forall P\geq P_{\star}.
\nonumber \end{align}
%Consider the case when $C(P)\geq H(X|Y).$ Then, there exists a $P_{\star}$ such that $C_{CR}(P_{\star})=H(X)$ for all $P$
%it holds that $C_{CR}(P)=H(X)$, where $H(X)$ is the best possible amount of common randomness that we can achieve with the available resources.
\begin{example}
\label{example1}
Consider the example of binary sources such that $|\mathcal{X}|=|\mathcal{Y}|=2$ with $P_{X}(0)=P_{X}(1)=\frac{1}{2}$. We consider the following transition probability 
\begin{align}
    P_{Y|X}(y|0)=\begin{pmatrix}1-\mu\\\mu\end{pmatrix} \quad 0\leq\mu\leq\frac{1}{2}
\nonumber \end{align}
\begin{align}
    P_{Y|X}(y|1)=\begin{pmatrix}\mu\\1-\mu\end{pmatrix} \quad 0\leq\mu\leq\frac{1}{2}
\nonumber \end{align}
with 
\begin{align}
    P_{XY}(x,y)=P_{Y|X}(y|x)P_{X}(x) \quad (x,y)\in\{0,1\}^2.
\nonumber \end{align}
In this case, it holds that

\begin{align}
    C_{CR}(P_{\star})=C_{CR}(P_{\star},P_{XY})=1
\nonumber \end{align}
and that 
\begin{align}
\frac{1}{2} \log(1+\frac{P_{\star}}{\sigma^{2}}) &= H(X|Y) \nonumber \\
   %  &= \underset{y \in \{0,1\}}{\sum}P_{Y}(y)\underset{x \in \{0,1\}}{\sum} P_{X|Y}(x|y) \log\left(\frac{1}{P_{X|Y}(x|y)}\right) \nonumber \\
     %&\overset{(a)}{=}\underset{y \in \{0,1\}}{\sum}\underset{x \in \{0,1\}}{\sum} P_{Y|X}(y|x) P_{X}(x) \log\left(\frac{P_{Y}(y)}{P_{Y|X}(y|x)P_{X}(x)}\right) \nonumber \\
     %&=\frac{1}{2} \underset{y \in \{0,1\}}{\sum} P_{Y|X}(y|0) \log\left(\frac{P_{Y}(y)}{\frac{1}{2}P_{Y|X}(y|0)}\right) \nonumber \\
     %&+\frac{1}{2} \underset{y \in \{0,1\}}{\sum} P_{Y|X}(y|1) \log\left(\frac{P_{Y}(y)}{\frac{1}{2}P_{Y|X}(y|1)}\right) \nonumber \\
     %&\overset{(b)}{=}\frac{1}{2} \underset{y \in \{0,1\}}{\sum} P_{Y|X}(y|0) \log\left(\frac{P_{Y|X}(y|0)+P_{Y|X}(y|1)}{P_{Y|X}(y|0)}\right) \nonumber \\
     %&+\frac{1}{2} \underset{y \in \{0,1\}}{\sum} P_{Y|X}(y|1) \log\left(\frac{P_{Y|X}(y|0)+P_{Y|X}(y|1)}{P_{Y|X}(y|1)}\right). \nonumber \\
     &=(1-\mu)\log(\frac{1}{1-\mu})+\mu\log(\frac{1}{\mu}). 
\nonumber \end{align}
We define 
\begin{align}
    f(\mu)=(1-\mu)\log(\frac{1}{1-\mu})+\mu\log(\frac{1}{\mu}).
\nonumber \end{align}
As a result, $P_{\star}$ is chosen such that:

\begin{align}
 P_{\star}=\sigma^{2}(2^{2f(\mu)}-1) \quad 0\leq\mu\leq\frac{1}{2}.
\nonumber \end{align}
In Fig. \ref{figplotbinary}, the channel power $P_{\star}$ is plotted as a function of the parameter $\mu$, with a fixed noise variance of $\sigma^2=1$. As $\mu$ increases, the correlation between the binary sources decreases. Consequently, the optimal power $P_{\star}$, starting at which the common randomness capacity is the highest possible, also increases, as depicted in Fig. \ref{figplotbinary}.
\begin{figure}[hbt!]
  \centering
\definecolor{mycolor1}{rgb}{0.00000,0.44700,0.74100}%
\scalebox{.9}{
\begin{tikzpicture}

\begin{axis}[%
width=3in,
height=3in,
at={(0in,0in)},
scale only axis,
xmin=0,
xmax=0.5,
xlabel style={font=\color{white!15!black}},
xlabel={$\mu$},
ymin=0,
ymax=3,
ylabel style={font=\color{white!15!black}},
ylabel={$P_{\star}$},
axis background/.style={fill=white},
xmajorgrids,
ymajorgrids,
legend style={legend cell align=left, align=left, draw=white!15!black}
]
\addplot [color=mycolor1]
  table[row sep=crcr]{%
0	0\\
0.01	0.118516293041571\\
0.02	0.216622073770917\\
0.03	0.309289123831429\\
0.04	0.399182721030622\\
0.05	0.487401275304271\\
0.06	0.574495722785659\\
0.07	0.66076425972601\\
0.08	0.746368944540166\\
0.09	0.831390741387116\\
0.1	0.915858733228839\\
0.11	0.999767000128063\\
0.12	1.08308503035076\\
0.13	1.16576449062496\\
0.14	1.24774382823797\\
0.15	1.32895152169362\\
0.16	1.40930845647063\\
0.17	1.48872971592012\\
0.18	1.56712597022471\\
0.19	1.64440458234272\\
0.2	1.72047051030039\\
0.21	1.79522706000599\\
0.22	1.8685765262986\\
0.23	1.94042074893613\\
0.24	2.01066160271362\\
0.25	2.079201435678\\
0.26	2.14594346571374\\
0.27	2.21079214312393\\
0.28	2.27365348490583\\
0.29	2.33443538500229\\
0.3	2.39304790375616\\
0.31	2.44940353900347\\
0.32	2.50341748064133\\
0.33	2.5550078500502\\
0.34	2.60409592539903\\
0.35	2.6506063535927\\
0.36	2.6944673494125\\
0.37	2.73561088223965\\
0.38	2.77397285062798\\
0.39	2.80949324489656\\
0.4	2.84211629784053\\
0.41	2.87179062360433\\
0.42	2.89846934472108\\
0.43	2.92211020729389\\
0.44	2.94267568427567\\
0.45	2.96013306679237\\
0.46	2.97445454344943\\
0.47	2.98561726756053\\
0.48	2.99360341224115\\
0.49	2.99840021331627\\
0.5	3\\
};

\end{axis}

\end{tikzpicture}}%
  \caption{Channel power $P_\star$ in function of the parameter $\mu$, for a noise variance $\sigma^2=1.$}
  \label{figplotbinary}
\end{figure}
%$(a)$ follows from the Bayes rule and $(b)$ follows from the law of total probability.\\
\end{example}
\end{remark}
\subsubsection{Optimization Problem:}
In this section, we solve the constrained optimization problem presented in Proposition $\ref{proposition1}$.
We consider the same sources and transition probability as in Example $\ref{example1}$ ($P_{X}$ and $P_{Y|X}$ are given as in  Example $\ref{example1}$).
Assume that the random variable $U$ has alphabet $\mathcal{U}$, then by applying the Support Lemma \cite{codingtheorems}, it holds that the cardinality of the set $\mathcal{U}$ satisfies the following constraint \cite{part2}
\begin{equation}
    |\mathcal{U}|\leq |\mathcal{X}|+1.
\nonumber \end{equation}
For $U\circlearrow{X}\circlearrow{Y}$, it holds that
\begin{equation}
    I(U;X)-I(U;Y)=I(U;X|Y),
\nonumber \end{equation}
where
\begin{equation}
    I(U;X|Y)\leq H(X|Y).
\nonumber \end{equation}
We can write the optimization problem in \eqref{CRcapacity formula} as follows:
\begin{equation}
C_{CR}(P)= \underset{\substack{U \\{\substack{U \circlearrow{X} \circlearrow{Y}\\ I(U;X|Y) \leq \min\{C(P),H(X|Y)\}}} \\ |\mathcal{U}|\leq |\mathcal{X}|+1 }}{\max} I(U;X).  
\nonumber \end{equation}

 Let $\mathcal{U}=\{u_{1},u_{2},u_{3}\}$.
We define $\boldsymbol{\theta}$ as follows:
\begin{equation}
    \boldsymbol{\theta}\nonumber \\               =\begin{pmatrix}P_{UX}\left(u_{1},0\right)\\P_{UX}\left(u_{2},0\right)\\P_{UX}\left(u_{3},0\right)\\P_{UX}\left(u_{1},1\right)\\P_{UX}\left(u_{2},1\right)\\P_{UX}\left(u_{3},1\right)\end{pmatrix},
\nonumber \end{equation}
with
\begin{equation}
\left( \begin{array}{rrrrrr}
1 & 1 & 1 & 0 & 0 & 0 \\
0 & 0 & 0 & 1 & 1 & 1 \\
\end{array}\right) 
 \boldsymbol{\theta}=\begin{pmatrix}P_{X}\left(0\right) \\ P_{X}\left(1\right)\end{pmatrix}.
\nonumber \end{equation}
\\~\\
We obtain the following equivalent constrained optimization problem:

\begin{equation}
\label{originalproblem}
    \underset{\boldsymbol{\theta}\in\Theta: \ \forall i \in \{1,\hdots 5\}, \ g_{i}\left(\boldsymbol{\theta}\right)\leq 0}{\min} g_{0}\left(\boldsymbol{\theta}\right),
\nonumber \end{equation}
where $\Theta=\{\boldsymbol{\theta}: \boldsymbol{\theta}\geq 0  \ \text{and} \ \boldsymbol{1}\boldsymbol{\theta}=1       \}$
with $\boldsymbol{1}=\left(1,1,1,1,1,1\right).$
\\~\\
The objective function is
\begin{align}
  g_{0}\left(\boldsymbol{\theta}\right)=-I(U;X). 
\nonumber \end{align}
In addition, the constraint functions are expressed as follows:
\begin{align}
   g_{1}\left(\boldsymbol{\theta}\right) 
   =I(U;X|Y)-\min\{C(P),H(X|Y)\} 
  % &=\underset{y\in\mathcal{Y}}{\sum}P_{Y}\left(y\right)\underset{u\in\mathcal{U},x\in\mathcal{X}}{\sum}P_{UX|Y}\left(u,x|y\right)\log\left(\frac{P_{UX|Y}\left(u,x|y\right)}{P_{U|Y}\left(u|y\right)P_{X|Y}\left(x|y\right)}\right)\nonumber \\
  % &\ -\min\{C(P),H(X|Y)\}.\nonumber \\
  % &\overset{(a)}{= }\underset{u\in\mathcal{U}x\in\mathcal{X}y\in\mathcal{Y}}{\sum}P_{Y|X}\left(y|x\right)P_{UX}\left(u,x\right)\log\left(\frac{\frac{P_{Y|X}\left(y|x\right)}{P_{X|Y}\left(x|y\right)}P_{UX}\left(u,x\right)}{\underset{x\in\mathcal{X}}{\sum}\left(P_{UX}\left(u,x\right)P_{Y|X}\left(y|x\right)\right)}\right) \nonumber \\
   %& \ -\min\{C(P),H(X|Y)\}.
\nonumber \end{align} 
\begin{equation}
   g_{2}\left(\boldsymbol{\theta}\right)=\underset{u\in \mathcal{U}}{\sum} P_{UX}\left(u,0\right)-P_{X}\left(0\right) 
\nonumber \end{equation}
\begin{equation}
   g_{3}\left(\boldsymbol{\theta}\right)=-\underset{u\in \mathcal{U}}{\sum} P_{UX}\left(u,0\right)+P_{X}\left(0\right) 
\nonumber \end{equation}
\begin{equation}
   g_{4}\left(\boldsymbol{\theta}\right)=\underset{u\in \mathcal{U}}{\sum} P_{UX}\left(u,1\right)-P_{X}\left(1\right) 
\nonumber \end{equation}
\begin{equation}
   g_{5}\left(\boldsymbol{\theta}\right)=-\underset{u\in \mathcal{U}}{\sum} P_{UX}\left(u,1\right)+P_{X}\left(1\right). 
\nonumber \end{equation}

\begin{remark} The optimization problem is non-convex since the objective function 
$g_{0}\left(\boldsymbol{\theta}\right)$ is non-convex. The non-convexity of $g_{0}\left(\boldsymbol{\theta}\right)$ is shown in the appendix.   
\end{remark}

%and 
%\begin{align}
 %   \frac{\partial^{2} g_{1}\left(\boldsymbol{\theta}\right)}{\partial^{2} P_{UX}\left(u,x\right)}=
  %  &\underset{y \in \mathcal{Y}}{\sum} P_{Y|X}\left(y|x\right)\frac{\underset{x \in \mathcal{X}}{\sum}P_{UX|Y}\left(u,x|y\right)- P_{UX|Y}\left(u,x|y\right) }{P_{UX|Y}\left(u,x|y\right)\underset{x \in \mathcal{X}}{\sum}P_{UX|Y}\left(u,x|y\right) }>0
%\end{align}

%\end{proof}

%\subsection{Lagrangian Formulation}
To solve the optimization problem, we define the Lagrangian function $\mathcal{L}:\Theta \times \Lambda \to \mathbb{R}$
\begin{equation}
    \mathcal{L}\left(\boldsymbol{\theta},\boldsymbol{\lambda}\right)=g_{0}(\boldsymbol{\theta})+\sum_{i=1}^{5} \lambda_{i} g_{i}\left(\boldsymbol{\theta}\right),
    \label{Lagrangian}
\end{equation}where $\Lambda=\{\boldsymbol{\lambda}=\left(\lambda_{1},\lambda_{2},\lambda_{3},\lambda_{4},\lambda_{5}\right)^{T} \in \mathbb{R}^{5}: \boldsymbol{\lambda} \geq 0\}. $ 

Optimizing the Lagrangian is interpreted as playing a two-player zero-sum game: the first player chooses $\boldsymbol{\theta}$  that minimizes $\mathcal{L}\left(\boldsymbol{\theta},\boldsymbol{\lambda}\right)$ 
and the second player chooses $\bs{\lambda}$ that maximizes it.
A pure Nash equilibrium might in general not exist. However, a mixed Nash Equilibrium does exist \cite{twoplayermethod}.
In what follows, the relationship between an approximate mixed
Nash equilibrium of the Lagrangian game and a nearly-optimal nearly-feasible solution to \eqref{originalproblem} is characterized.

\begin{theorem}\cite{twoplayermethod}
\label{feasiblesolution}
Let $\bs{\theta}^{\left(1\right)} \hdots \bs{\theta}^{\left(T\right)} \in \Theta$ and $\bs{\lambda}^{\left(1\right)}\hdots \bs{\lambda}^{\left(T\right)} \in \Lambda$ be sequences of vectors that satisfy an approximate mixed Nash equilibrium, i.e.
\begin{equation}
    \underset{\lambda^{\star} \in \Lambda}{\max}\frac{1}{T}\sum_{t=1}^{T}\mathcal{L}\left(\bs{\theta}^{\left(t\right)},\bs{\lambda}^{\star}\right)-\underset{\bs{\theta}^{\star}\in\Theta}{\inf} \frac{1}{T}\sum_{t=1}^{T}\mathcal{L}\left(\bs{\bs{\theta}}^{\star},\bs{\lambda}^{\left(t\right)}\right) \leq \epsilon.
\nonumber \end{equation}
Define $\Bar{\bs{\theta}}$ such that $\Bar{\bs{\theta}}=\bs{\theta}^{\left(t\right)}$ with probability $\frac{1}{T}$. Then it holds that $\Bar{\bs{\theta}}$ is nearly-optimal in expectation, i.e.,
\begin{equation}
    \mathbb{E}_{\Bar{\bs{\theta}}}\left[g_{0}\left(\Bar{\bs{\theta}}\right)\right]\leq \underset{\bs{\theta}^{\star}\in\Theta:\forall i: g_{i}\left(\bs{\theta}^{\star}\right)\leq 0}{\inf}g_{0}\left(\bs{\theta}^{\star}\right)+\epsilon.
\nonumber \end{equation}
\end{theorem}
%\subsection{Algorithm}
The following algorithm optimizes the Lagrangian function in \eqref{Lagrangian} in the non-convex setting.
\begin{algorithm}[H]
\caption{Lagrangian-formulation Optimization in the non-convex setting \cite{twoplayermethod}}
\label{nonconvexLagrangian}
\begin{enumerate}
\item Initialize $\boldsymbol{\lambda^{\left(1\right)}}=\boldsymbol{0}$
\item for $t \in [T]$
\begin{enumerate}
\item $\bs{\theta}^{\left(t\right)}=\mathcal{O}_{\rho}\left(\mathcal{L}\left(.,\bs{\lambda}^{\left(t\right)}\right)\right)$
\item  $\Delta_{\boldsymbol{\lambda}}^{\left(t\right)} $ gradient of $\mathcal{L}\left(\bs{\theta}^{\left(t\right)},\bs{\lambda}^{\left(t\right)}\right)$ w.r.t to $\bs{\lambda}$
\item Update $\boldsymbol{\lambda}^{\left(t+1\right)}=\Pi_{\Lambda}\left(\bs{\lambda}^{\left(t\right)}+\frac{\eta_{\lambda}}{\sqrt{\mathbf{G}_{\bs{\lambda},t}+\tau_{\lambda}}}\Delta_{\boldsymbol{\lambda}}^{\left(t\right)}\right)$
\end{enumerate}
\text{end} 
\item Return $\boldsymbol{\theta}^{\left(1\right)} \dots \boldsymbol{\theta}^{\left(T\right)}$ \quad $\boldsymbol{\lambda}^{\left(1\right)} \dots \boldsymbol{\lambda}^{\left(T\right)}$
\end{enumerate}
\end{algorithm}
The step a) consists of computing the  $\rho$-approximate Bayesian optimization oracle which is defined as follows:  
\begin{definition}\cite{twoplayermethod}
A $\rho$- approximate Bayesian optimization oracle is a function $\mathcal{O}_{\rho}:\left(\Theta \to \mathbb{R}\right) \to\Theta$ for which:
\begin{equation}
   f\left(\mathcal{O}_{\rho}\left(f\right)\right) \leq \underset{\bs{\theta}^{\star}\in\Theta}{\inf}f\left(\bs{\theta}^\star\right)+\rho.
\nonumber \end{equation}
\end{definition}
The gradient in step b) is expressed as follows:
\begin{align}
    \Delta_{\boldsymbol{\lambda}}= \left(g_{1}\left(\boldsymbol{\theta}\right),g_{2}\left(\boldsymbol{\theta}\right),g_{3}\left(\boldsymbol{\theta}\right),g_{4}\left(\boldsymbol{\theta}\right),g_{5}\left(\boldsymbol{\theta}\right)\right)^{T}.
\nonumber \end{align}
In step c), we perform first an AdaGrad \cite{AdaGrad} update, where $\eta_{\lambda}$ stands for the initial learning rate and $\tau_{\lambda}$ is a smoothing term.
In addition, $\mathbf{G}_{\bs{\lambda},t}$ is a diagonal matrix that contains the sum of the squares of the past gradients with respect to all parameters $\bs{\lambda}$ along its diagonal.
Second, we perform a projection onto $\Lambda$ such that:
\begin{equation}
\Pi_{\Lambda}\left(\bs{z}\right)=\max\left(\bs{0},\bs{z}\right).
\nonumber \end{equation}
We obtain $T$ candidate solutions $\boldsymbol{\theta}^{\left(1\right)} \dots \boldsymbol{\theta}^{\left(T\right)}$.
The goal is to yield a uniform distribution over these  $T$ candidates that is approximately feasible according to Theorem \ref{feasiblesolution}.
\begin{remark}
The approach proposed above is idealized \cite{twoplayermethod}. In practice, we opt for the typical approach:   pretending that our problem is convex and using a first-order stochastic algorithm such as AdaGrad . This is illustrated in Algorithm \ref{Lagrangianconvex}.
\end{remark}

\begin{algorithm}[H]
\caption{Lagrangian-formulation Optimization in the convex setting \cite{twoplayermethod}}
\label{Lagrangianconvex}
\begin{enumerate}
\item Initialize $\boldsymbol{\theta^{\left(1\right)}} \in \boldsymbol{\Theta}$ $\boldsymbol{\lambda^{\left(1\right)}}=\boldsymbol{0}$
\item for $t \in [T]$
\begin{enumerate}
\item $\check{\Delta}_{\boldsymbol{\theta}}^{\left(t\right)} $ sub-gradient of $\mathcal{L}\left(\boldsymbol{\theta}^{\left(t\right)},\boldsymbol{\lambda}^{\left(t\right)}\right)$ w.r.t to $\bs{\theta}$
\item  $\Delta_{\boldsymbol{\lambda}}^{\left(t\right)} $ gradient of $\mathcal{L}\left(\bs{\theta}^{\left(t\right)},\bs{\lambda}^{\left(t\right)}\right)$ w.r.t to $\bs{\lambda}$
 \item Update $\boldsymbol{\theta}^{\left(t+1\right)}=\Pi_{\Theta}\left(\bs{\theta}^{\left(t\right)}-\frac{\eta_{\theta}}{\sqrt{\mathbf{G}_{\bs{\theta},t}+\tau_{\theta}}}\Delta_{\boldsymbol{\theta}}^{\left(t\right)}\right)$
\item Update $\boldsymbol{\lambda}^{\left(t+1\right)}=\Pi_{\Lambda}\left(\bs{\lambda}^{\left(t\right)}+\frac{\eta_{\lambda}}{\sqrt{\mathbf{G}_{\bs{\lambda},t}+\tau_{\lambda}}}\Delta_{\boldsymbol{\lambda}}^{\left(t\right)}\right)$
\end{enumerate}
\text{end} 
\item Return $\boldsymbol{\theta}^{\left(1\right)} \dots \boldsymbol{\theta}^{\left(T\right)}$ \quad $\boldsymbol{\lambda}^{\left(1\right)} \dots \boldsymbol{\lambda}^{\left(T\right)}$
\end{enumerate}
\end{algorithm}
In step a), we compute the sub-gradient $\check{\Delta}_{\boldsymbol{\theta}}$, where it holds that \\

$\forall u \in \mathcal{U}:$
\begin{align}
 &\frac{\partial\mathcal{L}\left(\boldsymbol{\theta},\boldsymbol{\lambda}\right)}{\partial P_{UX}\left(u,0\right)}=-\log\left(\frac{P_{UX}\left(u,0\right)}{\underset{x'\in \mathcal{X}}{\sum}P_{UX}\left(u,x'\right)}\right) \nonumber \\ &+\lambda_{1}\underset{y\in\mathcal{Y}}{\sum} P_{Y|X}\left(y|0\right)\log\left(\frac{P_{Y|X}\left(y|0\right)P_{UX}\left(u,0\right)}{\underset{x'\in\mathcal{X}}{\sum}P_{Y|X}\left(y|x'\right)P_{UX}\left(u,x'\right)}\right) \nonumber \\
 &+\lambda_{2}-\lambda_{3} 
\nonumber \end{align}
%(\RM{1}) $\frac{\partial g_{0}\left(\boldsymbol{\theta}\right)}{\partial P_{UX}\left(u,x\right)}$ is computed  in $\eqref{derivg0}$ in the appendix.\\
and
\begin{align}
 &\frac{\partial\mathcal{L}\left(\boldsymbol{\theta},\boldsymbol{\lambda}\right)}{\partial P_{UX}\left(u,1\right)} \nonumber \\
 &=-\log\left(\frac{P_{UX}\left(u,1\right)}{\underset{x' \in \mathcal{X}}{\sum}P_{UX}\left(u,x'\right)}\right) \nonumber \\ &+\lambda_{1}\underset{y\in\mathcal{Y}}{\sum} P_{Y|X}\left(y|1\right)\log\left(\frac{P_{Y|X}\left(y|1\right)P_{UX}\left(u,1\right)}{\underset{x'\in\mathcal{X}}{\sum}P_{Y|X}\left(y|x\right)P_{UX}\left(u,x'\right)}\right) \nonumber \\
 &+\lambda_{4}-\lambda_{5}. 
\nonumber \end{align}
It is worth-mentioning here that $\frac{\partial g_{0}\left(\boldsymbol{\theta}\right)}{\partial P_{UX}\left(u,x\right)}$ is computed  in $\eqref{derivg0}$ and  that $\frac{\partial g_{1}\left(\boldsymbol{\theta}\right)}{\partial P_{UX}\left(u,x\right)}$ is computed for $U\circlearrow{X}\circlearrow{Y}$ in  $\eqref{derivg1estimate}$ (see Appendix).\\~\\
In step c) and step d), AdaGrad updates are performed, where $\eta_{\lambda}$ and $\eta_{\theta}$ correspond to  the initial learning rates and $\tau_{\lambda}$ and $\tau_{\theta}$ are smoothing terms.
Furthermore, $\mathbf{G}_{\bs{\lambda},t}$ and $\mathbf{G}_{\bs{\theta},t}$ are diagonal matrices that contain  the sum of the squares of the past gradients with respect to all parameters $\bs{\lambda}$ and $\bs{\theta}$ respectively, along their diagonal.

The projection $\Pi_{\Theta}\left(\bs{z}\right)$ onto $\Theta$ in step c) corresponds to the Euclidean projection onto the probability simplex which is computed using the following algorithm:

\begin{algorithm}[H]
\caption{Euclidean projection of a vector onto the probability simplex \cite{projectionsimplex}}
\begin{enumerate}
\item \textbf{Input:} $\boldsymbol{z} \in \mathbb{R}^{D}$ 
\begin{enumerate}
\item Sort $\boldsymbol{z}$ into $\boldsymbol{w}:w_{1}\geq w_{2} \geq \hdots \geq w_{D}$
\item Find $\gamma=\max\left(1\leq j \leq D: w_{j}+\frac{1-\sum_{i=1}^{j} w_{i}}{j}\right)$
\item Define $\kappa=\frac{1}{\gamma}\left(1-\sum_{i=1}^{\gamma}w_{i}\right)$
\end{enumerate}
\item $\textbf{Output:}\  \boldsymbol{\theta} \ \text{s.t} \ \theta_{i}=\max\{z_{i}+\kappa,0\} \  i=1\hdots D $
\end{enumerate}
\end{algorithm}
\subsubsection{Simulation Results:}
In this section, we present our numerical results. We study the CR capacity for different channel input powers as well as for different values of the parameter $\mu$. We fix $T=5000$.
  Algorithm $\ref{Lagrangianconvex}$ is implemented for given $P$ and given $\mu$ and for different values of the initial learning rates $\eta_{\theta}$ and $\eta_{\lambda}$.
  At the end, we consider the pair $\left(\eta_{\theta},\eta_{\lambda}\right)$ for which $\bs{\theta}^{\left(1\right)} \hdots \bs{\theta}^{\left(T\right)} \in \Theta$ and $\bs{\lambda}^{\left(1\right)}\hdots \bs{\lambda}^{\left(T\right)} \in \Lambda$ yield the smallest $\epsilon$ in Theorem \ref{feasiblesolution}.
  We vary first the parameter $\mu$ in $\{0,0.1,0.2,0.3,0.4,0.5\}$ and plot for each $\mu$ the common randomness capacity as a function of the power, as depicted in Fig. \ref{figplotmu}.
\begin{figure}[!t]
  \centering
\definecolor{mycolor1}{rgb}{0.00000,0.44700,0.74100}%
\definecolor{mycolor2}{rgb}{1.00000,0.00000,1.00000}%
\definecolor{mycolor3}{rgb}{0.49412,0.18431,0.55686}%
\scalebox{.64}{\begin{tikzpicture}

\begin{axis}[%
width=4.521in,
height=3.57in,
at={(0.758in,0.482in)},
scale only axis,
xmin=0,
xmax=3.3,
xlabel style={font=\color{white!15!black}},
xlabel={\textbf{Power}},
ymin=0,
ymax=1.01,
ylabel style={font=\color{white!15!black}},
ylabel={\textbf{Common Randomness Capacity}},
axis background/.style={fill=white},
xmajorgrids,
ymajorgrids,
legend style={at={(0.681,0.368)}, anchor=south west, legend cell align=left, align=left, draw=white!15!black}
]
\addplot [color=mycolor1, mark=o, mark options={solid, mycolor1}]
  table[row sep=crcr]{%
0	0.999999999999956\\
0.4	0.999999999999956\\
0.443700637652136	0.999999999999956\\
0.487401275304271	0.999999999999956\\
0.815858733228839	0.999999999999956\\
0.865858733228839	0.999999999999956\\
0.915858733228839	0.999999999999956\\
1.25	0.999999999999956\\
1.28947576084681	0.999999999999956\\
1.32895152169362	0.999999999999956\\
1.62047051030039	0.999999999999956\\
1.67047051030039	0.999999999999956\\
1.72047051030039	0.999999999999956\\
2	0.999999999999956\\
2.039600717839	0.999999999999956\\
2.079201435678	0.999999999999956\\
2.29304790375616	0.999999999999956\\
2.34304790375616	0.999999999999956\\
2.39304790375616	0.999999999999956\\
2.6	0.999999999999956\\
2.62530317679635	0.999999999999956\\
2.6506063535927	0.999999999999956\\
2.74211629784053	0.999999999999956\\
2.79211629784053	0.999999999999956\\
2.84211629784053	0.999999999999956\\
2.9	0.999999999999956\\
2.93006653339619	0.999999999999956\\
2.96013306679237	0.999999999999956\\
2.98006653339619	0.999999999999956\\
3	0.999999999999956\\
3.1	0.999999999999956\\
3.2	0.999999999999956\\
3.3	0.999999999999956\\
};
\addlegendentry{$\mu=0$}

\addplot [color=red, mark=asterisk, mark options={solid, red}]
  table[row sep=crcr]{%
0	0\\
0.4	0.588559918279829\\
0.443700637652136	0.637623035994095\\
0.487401275304271	0.686527722945682\\
0.815858733228839	0.937647068777706\\
0.865858733228839	0.967702057356215\\
0.915858733228839	0.999999999999956\\
1.25	0.999999999999956\\
1.28947576084681	0.999999999999956\\
1.32895152169362	0.999999999999956\\
1.62047051030039	0.999999999999956\\
1.67047051030039	0.999999999999956\\
1.72047051030039	0.999999999999956\\
2	0.999999999999956\\
2.039600717839	0.999999999999956\\
2.079201435678	0.999999999999956\\
2.29304790375616	0.999999999999956\\
2.34304790375616	0.999999999999956\\
2.39304790375616	0.999999999999956\\
2.6	0.999999999999956\\
2.62530317679635	0.999999999999956\\
2.6506063535927	0.999999999999956\\
2.74211629784053	0.999999999999956\\
2.79211629784053	0.999999999999956\\
2.84211629784053	0.999999999999956\\
2.9	0.999999999999956\\
2.93006653339619	0.999999999999956\\
2.96013306679237	0.999999999999956\\
2.98006653339619	0.999999999999956\\
3	0.999999999999956\\
3.1	0.999999999999956\\
3.2	0.999999999999956\\
3.3	0.999999999999956\\
};
\addlegendentry{$\mu=0.1$}

\addplot [color=mycolor2, mark=square, mark options={solid, mycolor2}]
  table[row sep=crcr]{%
0	0\\
0.4	0.356715966801272\\
0.443700637652136	0.386983717545609\\
0.487401275304271	0.428427101978334\\
0.815858733228839	0.631816959696426\\
0.865858733228839	0.658549119279064\\
0.915858733228839	0.683883457664767\\
1.25	0.83559452771623\\
1.28947576084681	0.843763551865788\\
1.32895152169362	0.857478655511383\\
1.62047051030039	0.963165181222168\\
1.67047051030039	0.979631387926287\\
1.72047051030039	0.999999999999956\\
2	0.999999999999956\\
2.039600717839	0.999999999999956\\
2.079201435678	0.999999999999956\\
2.29304790375616	0.999999999999956\\
2.34304790375616	0.999999999999956\\
2.39304790375616	0.999999999999956\\
2.6	0.999999999999956\\
2.62530317679635	0.999999999999956\\
2.6506063535927	0.999999999999956\\
2.74211629784053	0.999999999999956\\
2.79211629784053	0.999999999999956\\
2.84211629784053	0.999999999999956\\
2.9	0.999999999999956\\
2.93006653339619	0.999999999999956\\
2.96013306679237	0.999999999999956\\
2.98006653339619	0.999999999999956\\
3	0.999999999999956\\
3.1	0.999999999999956\\
3.2	0.999999999999956\\
3.3	0.999999999999956\\
};
\addlegendentry{$\mu=0.2$}

\addplot [color=black, mark=diamond, mark options={solid, black}]
  table[row sep=crcr]{%
0	0\\
0.4	0.281817519041181\\
0.443700637652136	0.306438130873813\\
0.487401275304271	0.332024513242334\\
0.815858733228839	0.493966808169632\\
0.865858733228839	0.51539622052325\\
0.915858733228839	0.546315220239446\\
1.25	0.678089473689647\\
1.28947576084681	0.691981917457754\\
1.32895152169362	0.70578165286457\\
1.62047051030039	0.791303252027284\\
1.67047051030039	0.814731383542805\\
1.72047051030039	0.822057248705955\\
2	0.902119044114369\\
2.039600717839	0.913954214241758\\
2.079201435678	0.924869321919285\\
2.29304790375616	0.977052475715108\\
2.34304790375616	0.983699130001044\\
2.39304790375616	0.999999999999956\\
2.6	0.999999999999956\\
2.62530317679635	0.999999999999956\\
2.6506063535927	0.999999999999956\\
2.74211629784053	0.999999999999956\\
2.79211629784053	0.999999999999956\\
2.84211629784053	0.999999999999956\\
2.9	0.999999999999956\\
2.93006653339619	0.999999999999956\\
2.96013306679237	0.999999999999956\\
2.98006653339619	0.999999999999956\\
3	0.999999999999956\\
3.1	0.999999999999956\\
3.2	0.999999999999956\\
3.3	0.999999999999956\\
};
\addlegendentry{$\mu=0.3$}

\addplot [color=blue, mark=x, mark options={solid, blue}]
  table[row sep=crcr]{%
0	0\\
0.4	0.250774308425682\\
0.443700637652136	0.273432344381437\\
0.487401275304271	0.294864843309025\\
0.815858733228839	0.444000181244402\\
0.865858733228839	0.465749968724255\\
0.915858733228839	0.486409438391404\\
1.25	0.605092368278683\\
1.28947576084681	0.617731166123919\\
1.32895152169362	0.631301991695333\\
1.62047051030039	0.715931649892729\\
1.67047051030039	0.734142684993215\\
1.72047051030039	0.747079910465086\\
2	0.818608332157071\\
2.039600717839	0.828376036255501\\
2.079201435678	0.836081819090661\\
2.29304790375616	0.884489397155364\\
2.34304790375616	0.898295660341986\\
2.39304790375616	0.910082281116836\\
2.6	0.947506531298494\\
2.62530317679635	0.955855794786827\\
2.6506063535927	0.961520027432538\\
2.74211629784053	0.975788182321122\\
2.79211629784053	0.985510573540767\\
2.84211629784053	0.999999999999956\\
2.9	0.999999999999956\\
2.93006653339619	0.999999999999956\\
2.96013306679237	0.999999999999956\\
2.98006653339619	0.999999999999956\\
3	0.999999999999956\\
3.1	0.999999999999956\\
3.2	0.999999999999956\\
3.3	0.999999999999956\\
};
\addlegendentry{$\mu=0.4$}

\addplot [color=mycolor3, mark=triangle, mark options={solid, rotate=180, mycolor3}]
  table[row sep=crcr]{%
0	0\\
0.4	0.242133296905264\\
0.443700637652136	0.26406630084337\\
0.487401275304271	0.284904734068888\\
0.815858733228839	0.430238646456722\\
0.865858733228839	0.4495433628644\\
0.915858733228839	0.468789898026524\\
1.25	0.585073430831138\\
1.28947576084681	0.597418179335421\\
1.32895152169362	0.610034767527131\\
1.62047051030039	0.696302873659515\\
1.67047051030039	0.7098956058004\\
1.72047051030039	0.721529439300005\\
2	0.792560026699011\\
2.039600717839	0.803234637633183\\
2.079201435678	0.80905076091569\\
2.29304790375616	0.860184437220945\\
2.34304790375616	0.8708345776457\\
2.39304790375616	0.877549134656064\\
2.6	0.924088944526301\\
2.62530317679635	0.929445327381909\\
2.6506063535927	0.93471862530234\\
2.74211629784053	0.9511237332176\\
2.79211629784053	0.956696216850588\\
2.84211629784053	0.968357202120355\\
2.9	0.976814274877637\\
2.93006653339619	0.982326786340032\\
2.96013306679237	0.987782192681397\\
2.98006653339619	0.991330400411299\\
3	0.999999999999956\\
3.1	0.999999999999956\\
3.2	0.999999999999956\\
3.3	0.999999999999956\\
};
\addlegendentry{$\mu=0.5$}

\end{axis}

\begin{axis}[%
width=5.833in,
height=4.381in,
at={(0in,0in)},
scale only axis,
xmin=0,
xmax=1,
ymin=0,
ymax=1,
axis line style={draw=none},
ticks=none,
axis x line*=bottom,
axis y line*=left,
legend style={legend cell align=left, align=left, draw=white!15!black}
]
\end{axis}
\end{tikzpicture}}%
  \caption{CR capacity in function of the power for a noise variance $\sigma^2=1$ and for different values of $\mu.$}
  \label{figplotmu}
\end{figure}
Next, we generate a three-dimensional plot of the common randomness capacity as a function of both the power and the parameter $\mu$. This is illustrated in Fig. \ref{figplot3D}.
\begin{figure}[!t]
 \centering{\includegraphics[scale=0.42]{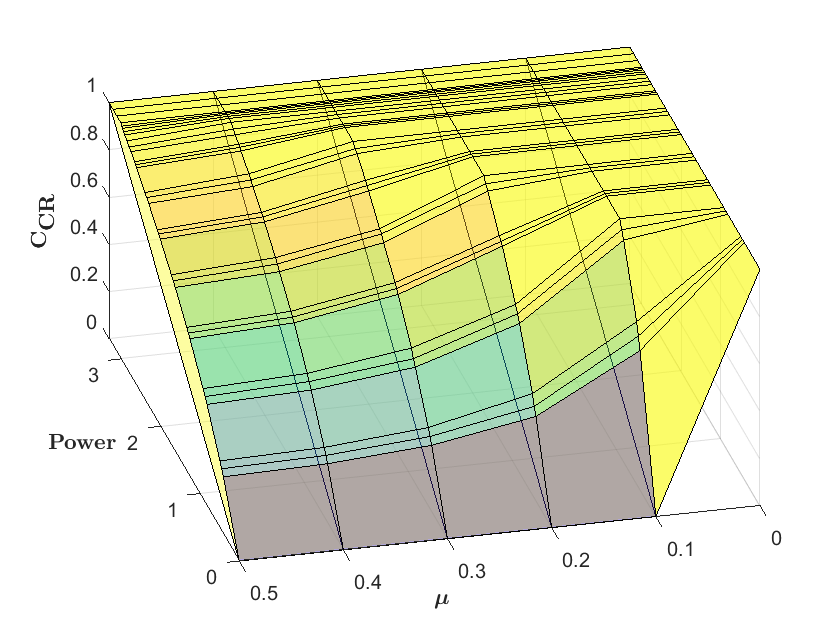}}
 %\centering{\includegraphics[scale=0.62]{figures/results3Dupdated2.eps}}
 \caption{CR capacity in function of the power and the parameter $\mu$ for a noise variance $\sigma^2=1.$}
 \label{figplot3D} 
\end{figure}
We consider first the specific case when $\mu=0$. Clearly, for $\mu=0$, it holds that $Y=X$. Therefore, no communication over the channel is required to achieve the maximal amount of common randomness equal to $H(X)$. This is numerically verified in Fig. \ref{figplotmu} and Fig. $\ref{figplot3D}$, where for $\mu=0$ and $P=0$ the common randomness capacity is equal to $H(X)=1.$
However, for $0<\mu\leq \frac{1}{2}$, a communication over the channel is necessary to generate common randomness between the two terminals $(C_{CR}\left(0\right)=0$ for $0<\mu\leq \frac{1}{2}$). This is due to the fact that $X$ and $Y$ have an indecomposable joint distribution \cite{part2} for $0<\mu\leq \frac{1}{2}$.
Furthermore, the higher the parameter $\mu$ is, the less correlated the sources are. As a result, we need to investigate more power in order to achieve the same amount of common randomness obtained for lower values of 
$\mu$. This is clearly observable in Fig. \ref{figplotmu}. For instance, for $\mu=0.2$, the common randomness capacity is equal to 0.86 for $P\thickapprox 1.33$, whereas, for $\mu=0.5$, the same amount of common randomness is achieved for $P \thickapprox 2.34$.
\subsection{MIMO Case}
We now focus on the second  scenario depicted in Fig. \ref{correlatedMIMO}, where the communication is over a  MIMO Gaussian channel with the power constraint defined in $\eqref{energyconstraintMIMOCorrelated}$.
\begin{proposition} For the model in Fig. \ref{correlatedMIMO}, the CR capacity $C_{CR}(P,N_{T} \times N_{R})$ is equal:
$$C_{CR}(P,N_{T} \times N_{R})= \underset{ \substack{U \\{\substack{U \circlearrow{X} \circlearrow{Y}\\ I(U;X)-I(U;Y) \leq C(P,N_{T} \times N_{R})}}}}{\max} I(U;X).  $$
\label{proposition2}
\end{proposition}
\subsubsection{Direct Proof}
\label{signalprocessing}
\begin{figure}
	\centering 
	\tikzstyle{vecArrow} = [thick, decoration={markings,mark=at position
   1 with {\arrow[semithick]{open triangle 60}}},
   double distance=1.4pt, shorten >= 5.5pt,
   preaction = {decorate},
   postaction = {draw,line width=1.4pt, white,shorten >= 4.5pt}]
\tikzstyle{block} = [draw, top color=white, bottom color=white!80!gray, rectangle, rounded corners,
minimum height=2em, minimum width=2cm]
\tikzstyle{blockSVD} = [draw, top color=white, bottom color=white!80!gray, rectangle, rounded corners,
minimum height=3cm, minimum width=.9cm]
\tikzstyle{input} = [coordinate]
\tikzstyle{sum} = [draw, circle,inner sep=1pt, minimum size=2mm, thick]
\usetikzlibrary{arrows}
\scalebox{.9}{
\begin{tikzpicture}[scale= 1,font=\footnotesize]
\node[] (xtilde) {\large$\boldsymbol{\tilde{T}}$};
\node[blockSVD,right=1cm of xtilde] (prep) {\large$\mathbf{V}$};
\node[blockSVD, right=1.5cm of prep](prepp) {\large
$\mathbf{V}^\h$};
\node[sum,right= .8cm of prepp.425] (lambda1) {$\times$};
\node[sum,right= .8cm of prepp.295] (lambda2) {$\times$};
%\node[sum, right=1cm of lambda1] (ntilde1) {$+$};
%\node[sum, right=1cm of lambda2] (ntilde2) {$+$};
\node[blockSVD, right=2cm of prepp] (postp) {\large$\mathbf{U}$};
\node[blockSVD, right=2.8cm of postp] (post) {\large$\mathbf{U}^\h$};
\node[right=1cm of post] (ytilde) {\large$ \boldsymbol{\tilde{Z}}$};
\node[above=.3cm of lambda1] (lamda1) {\large$\lambda_1$};
\node[above=.3cm of lambda2] (lamda2) {\large$\lambda_{N_{\text{min}}}$};
\node[sum,right=1cm of postp] (ksi) {$+$};
\node[above=.3cm of ksi] (ksi1){\large$ \boldsymbol{\tilde{\xi}}$};
%\node[above=.3cm of ntilde1] (ntild1) {\large${\tilde{\xi}_1}$};
%\node[above=.3cm of ntilde2] (ntild2) {\large${\tilde{\xi}_N}$};
\draw[vecArrow] (xtilde) -- (prep) ;
\draw[vecArrow] (prep) --node[above]{\large$\boldsymbol{T}$} (prepp);
\draw[->,thick] (prepp.425)--(lambda1);
\draw[->,thick] (prepp.295)--(lambda2);
\draw[->,thick] (lambda1)--(postp.115);
%\draw[->,thick] (ntilde1) --;
\draw[->,thick] (lambda2)--(postp.245);
%\draw[->,thick] (ntilde2)--;
\draw[vecArrow] (postp) -- (ksi);
\draw[vecArrow] (ksi) --node[above]{\large$\boldsymbol{Z}$} (post);
%node[above]{\large$\boldsymbol{y}$} (post);
\draw[vecArrow] (post) -- (ytilde);
\draw[->,thick] (lamda1) -- (lambda1);
\draw[->,thick] (lamda2) -- (lambda2);
\draw[->,thick] (ksi1) -- (ksi);
%\draw[->,thick] (ntild1) -- (ntilde1);
%\draw[->,thick] (ntild2) -- (ntilde2);
\draw[dashed] (3.3,-2.3) rectangle (9.3,2.3);
\node[below=.5 of prep] (precoder) { pre-processing};
\node[below=.5 of post] (decoder) { post-processing};
\node[below=0.01 cm of lambda1] (inf1) {$ \vdots$};
%\node[below=0.01 cm of ntilde1] (inf2) {$ \vdots$};
\node[] at (6, -2.5) {channel};

\end{tikzpicture}}
	\caption{Decomposition of the MIMO channel into $N_{\text{min}}$ parallel channels through SVD.}
	\label{Fig:SVD_Dec}
\end{figure}
%In this section, we provide optimal signal processing in the sense that with this processing we can prove the achievability of the common randomness capacity over MIMO Gaussian channels.
\begin{proof}
The capacity $C(P,N_T\times N_R)$ can be computed by converting the MIMO channel into parallel, independent and scalar Gaussian sub-channels. This conversion is based on the following singular value decomposition (SVD) of the channel matrix $\mathbf{H}$:
	\begin{equation}
	   \mathbf{H}=\mathbf{U}\mathbf{\Lambda}\mathbf{V}^\h,
	\nonumber \end{equation}
where $\mathbf{U} \in \mathbb{C}^{N_R\times N_R}$ and $\mathbf{V} \in \mathbb{C}^{N_T\times N_T}$ are unitary matrices. $\mathbf{\Lambda} \in \mathbb{C}^{N_R\times N_T} $ is a diagonal matrix, whose diagonal elements $\lambda_1\geq \lambda_2 \geq \cdots \geq \lambda_{N_{\text{min}}}$ are the ordered singular values of the channel matrix $\mathbf{H}$. We denote with $N_{\text{min}}$ the rank of $\mathbf{H}$, $N_{\text{min}} \coloneqq  \min(N_T,N_R)$. If we multiply \eqref{Model:MIMOcorrelated} with the unitary matrix $\mathbf{U}^\h$, we then obtain
\begin{equation}
\underbrace{\mathbf{U}^\h \bs{Z}}_{\coloneqq \tilde{\bs{Z}}}= \mathbf{U}^\h\mathbf{U}\mathbf{\Lambda}\underbrace{\mathbf{V}^\h \bs{T}}_{\coloneqq  \tilde{\bs{T}}}+\underbrace{\mathbf{U}^\h\bs{\xi}}_{\coloneqq \tilde{\bs{\xi}}}.
\nonumber \end{equation}
%Since $\mathbf{U}$ is a unitary matrix, \eqref{Model:MIMO} can be rewritten as follows:
%\begin{equation}
 %    \tilde{\bs{y}}= \mathbf{\Lambda} \tilde{\bs{x}} + \tilde{\bs{\xi}}
%\end{equation}
It can easily be checked that $\tilde{\bs{\xi}}$ has the same distribution as $\bs{\xi}$ \cite{teletar}, i.e.,  $\tilde{\bs{\xi}} \sim \mathcal{N}_{\mathbb{C}}(\bs{0}_{N_R},\sigma^2 \mathbf{I}_{N_R}),$ and we have
\begin{equation*}
\mathbb{E}[\tilde{\bs{T}}^\h\tilde{\bs{T}}]= \mathbb{E}[ \bs{T}^\h \mathbf{V} \mathbf{V}^\h \bs{T}]=  \mathbb{E}[\bs{T}^\h \bs{T}].
\end{equation*}
We obtain the $N_{\text{min}}$ independent scalar Gaussian channels depicted in Fig. \ref{Fig:SVD_Dec}
\begin{align}
    \tilde{Z}_{\ell}=\lambda_{\ell}\tilde{T}_{\ell}+\tilde{\xi}_{\ell} \ \ \ell=1\dots N_{\text{min}}.
\nonumber \end{align}
The SVD can be interpreted as a pre-processing (multiplication
with $\mathbf{V}$) and a post-processing (multiplication with $\mathbf{U}^{H}$).
The optimization problem in \eqref{transmissioncapacityMIMOcorrelated} is reduced to \cite{Tse}
\begin{align}
    &C(P,N_{T}\times N_{R})=\max_{\tilde{P}_{1}\dots\tilde{P}_{N_{\text{min}}}} \sum_{\ell=1}^{N_{\text{min}}} \log\left( 1+\frac{\lambda_{\ell}^{2}}{\sigma^{2}}\tilde{P_{\ell}}   \right), \nonumber \\
   & \text{s.t.} \sum_{\ell=1}^{N_{\text{min}}}\tilde{P_{\ell}} \leq P \  \text{and} \  \tilde{P_{\ell}}\geq 0 \ \  \ell=1 \dots N_{\text{min}}.
    \nonumber \end{align}
    It holds that
    \begin{align}
        C(P,N_{T}\times N_{R})=\sum_{\ell=1}^{N_{\text{min}}} C(\tilde{P}_{\ell}),
    \nonumber \end{align}
    where the capacity of each sub-channel is expressed as
    \begin{align}
       C(\tilde{P}_{\ell})= \log\left( 1+\frac{\lambda_{\ell}^{2}}{\sigma^{2}}\tilde{P_{\ell}}   \right) \ \  \ell=1 \dots N_{\text{min}}.
    \nonumber \end{align}
  The power  $\tilde{P_{\ell}}$ is called the \textit{waterfilling} rule \cite{waterfilling} and it is expressed as follows:
  \begin{equation}
      \tilde{P_{\ell}}=\max\left(0,\kappa-\frac{\sigma^{2}}{\lambda_{\ell}}\right) \ \ell=1 \dots N_{\text{min}},
  \nonumber \end{equation}
 where $\kappa$ is the \textit{waterfilling} level.

 By Proposition \ref{proposition1}, we have for $\ell=1\hdots N_{\text{min}}.$
\begin{align}
C_{CR}(\tilde{P_{\ell}})=\underset{\substack{U_\ell\\{\substack{U_\ell \circlearrow{X} \circlearrow{Y}\\ I(U_\ell;X)-I(U_\ell;Y) \leq C(\tilde{P}_{\ell})}}}}{\max} I(U_\ell;X). 
\nonumber \end{align}
Since we may lose information through processing, it holds that
\begin{align}
  &C_{CR}(P,N_{T}\times N_{R}) \nonumber \\ &\geq \sum_{\ell=1}^{N_{\text{min}}} C_{CR}(\tilde{P}_{\ell}) \nonumber\\
  &=\sum_{\ell=1}^{N_{\text{min}}}\underset{\substack{U_{\ell}\\{\substack{U_{\ell} \circlearrow{X} \circlearrow{Y}\\ I(U_{\ell};X)-I(U_{\ell};Y) \leq C(\tilde{P}_{\ell})}}}}{\max} I(U_{\ell};X) \nonumber \\
  &=\sum_{\ell=1}^{N_{\text{min}}}\underset{\substack{U_{1}\hdots U_{N_{\text{min}}}\\{\substack{U_{\ell} \circlearrow{X} \circlearrow{Y} \ \ell=1\hdots N_{\text{min}}\\ I(U_{\ell};X)-I(U_{\ell};Y) \leq C(\tilde{P}_{\ell}) \ \ell=1\hdots N_{\text{min}} }}}}{\max} I(U_{\ell};X) \nonumber \\
  &\geq \underset{\substack{U_{1}\hdots U_{N_{\text{min}}}\\{\substack{U_{\ell} \circlearrow{X} \circlearrow{Y} \ \ell=1\hdots N_{\text{min}}\\ I(U_{\ell};X)-I(U_{\ell};Y) \leq C(\tilde{P}_{\ell}) \ \ell=1\hdots N_{\text{min}} }}}}{\max} \sum_{\ell=1}^{N_{\text{min}}} I(U_{\ell};X) \nonumber \\
  &\geq \underset{\substack{U_{1}\hdots U_{N_{\text{min}}}\\{\substack{ I(U_{\ell};X)-I(U_{\ell};Y) \leq C(\tilde{P}_{\ell}) \ \ell=1\hdots N_{\text{min}} }} \\  U_{\ell} \ \text{independent of} \ (X,Y) ,\ \ell=1\hdots,N_{\text{min}} \\ U_{\ell},\ \ell=1\hdots,N_{\text{min}},\ \text{pairwise independent} \\ \sum_{\ell=1}^{N_{\text{min}}} I(U_{\ell};X)- \sum_{\ell=1}^{N_{\text{min}}} I(U_{\ell};Y) \leq \sum_{\ell=1}^{N_{\text{min}}} C(\tilde{P}_{\ell})}}{\max} \sum_{\ell=1}^{N_{\text{min}}} I(U_{\ell};X) \nonumber \\
     &\overset{(a)}{=}\underset{\substack{V\\{\substack{V \circlearrow{X} \circlearrow{Y}\\ I(V;X)-I(V;Y) \leq \sum_{\ell=1}^{N_{\text{min}}} C(\tilde{P}_{\ell})}}}}{\max} I(V;X) \nonumber \\
      &=\underset{\substack{V\\{\substack{V \circlearrow{X} \circlearrow{Y}\\ I(V;X)-I(V;Y) \leq C(P,N_{T}\times N_{R})}}}}{\max} I(V;X), 
      \label{existV}
\end{align}
where $(a)$ follows from defining $V$ such that $U_{\ell}\circlearrow{V}\circlearrow{X}\circlearrow{Y}, \ \ell=1\hdots N_{\text{min}}$, with $I(V;X)=\sum_{\ell=1}^{N_{\text{min}}}I(U_{\ell};X)$ and $I(V;Y)=\sum_{\ell=1}^{N_{\text{min}}}I(U_{\ell};Y)$. Here, $U_{\ell}$, $\ell=1\hdots N_{\text{min}}$, satisfy the following constraints:
 \begin{enumerate}
     \item  Each $U_{\ell}, \ \ell=1\hdots N_{\text{min}}$ is independent of $(X,Y).$
     \item $I(U_{\ell};X)-I(U_{\ell};Y) \leq C(\tilde{P}_{\ell}) \quad \text{for} \ \ell=1\hdots N_{\text{min}}.$
     \item $U_{\ell}$, $\ell=1\hdots N_{\text{min}}$, are pairwise independent.
 \end{enumerate}
 The existence of such $V$ is proved in the appendix.
\end{proof}
\subsubsection{Converse Part:}The converse proof for the MIMO case is analogous to the converse proof for the SISO case.
\begin{remark}
The signal processing presented in Section \ref{signalprocessing} is \textit{optimal} in the sense that with this processing, one can demonstrate the achievability of the common randomness capacity over the MIMO Gaussian channel.
\end{remark}
\begin{remark}
Since $C(P,N_{T}\times N_{R})\geq C(P)$, it follows from Proposition \ref{proposition1} and Proposition \ref{proposition2} that:
\begin{align}
C_{CR}(P,N_{T}\times N_{R})\geq C_{CR}(P).
\nonumber \end{align}
Intuitively, because the MIMO channel has a higher capacity than the SISO channel, the amount of information that can be reliably transmitted is greater. Consequently, by communicating over the MIMO channel, Terminals $A$ and $B$ can generate a greater amount of common randomness.
\end{remark}
\section{Application of Common Randomness: Secure Identification}
\label{sec4}
%\subsubsection{Randomized Identification-Code for the Gaussian Wiretap Channel}
In this section, we explore a significant application of CR generation: the identification paradigm. Unlike transmission, it appears that the resource CR can enhance the identification capacity of channels. We introduce a coding scheme for CR-assisted secure identification and prove a lower bound on the secure identification capacity within this setting, as illustrated in Proposition \ref{proposition3}.

\begin{proposition}
Let $C_{SID}^{c}(g,g^{\prime},P)$ and  $C_{S}(g,g^{\prime},P)$, respectively, be the secure identification capacity and the secrecy capacity for the model in Fig. \ref{fig:secureId}, respectively. It holds that

$ \text{if} \  C_{S}(g,g^{\prime},P)>0 \  \text{then}$
$$C_{SID}^{c}(g,g^{\prime},P) \geq \underset{\substack{U\\{\substack{U \circlearrow{X} \circlearrow{Y}\\ I(U;X)-I(U;Y) \leq C(g,P)}}}}{\max} I(U;X).  $$  
\label{proposition3}
\end{proposition}
\begin{remark}
Consider in particular the case when $C(g,P)\geq H(X|Y)$. Then, it follows from Proposition \ref{proposition3} that $C_{SID}^{c}(g,g^{\prime},P) \geq H(X)$ as long as $C_{S}(g,g^{\prime},P)>0$.
\end{remark}
\begin{proof}
%\begin{figure}[hbt!]
%\centering
%    \input{figures/wiretap.tex}
%   \caption{Wiretap channel with correlation}
%    \label{fig:wiretap}
%\end{figure}

 %One can particularly concentrate on secure identification. In this case, for the common randomness generation, it suffices to consider the following encoding scheme.
 
Given a DMMS $P_{XY}$, Alice observes the outputs $X^{n}$ and Bob observes the outputs $Y^{n}.$
 Alice generates a random variable $K$ with alphabet $\mathcal{K}=\{1\dots M^{\prime}\},$ such that $K=\Phi(X^{n}).$ To send a message $i$, we prepare a set of coloring-functions or mappings $E_i$ known by the sender and the receiver.\begin{align*}
 E_i  & \colon \mathcal{K} \longrightarrow \{1,\ldots,M^{\prime\prime} \} \\
 & \colon \underbrace{K}_{\text{coloring}} \mapsto \underbrace{E_i(K).}_{\text{color}}
 \end{align*}
 $X^{n}$ is encoded to a sequence $T^{n}$ using an error correcting code, then 
 %$K$ is encoded to a sequence  $\bs{u^{\prime}}$ of length $n_{0}=n-\sqrt{n}$, using a transmission code $\Cp=\{(\up_K,\dcp_K),K \in \{1,\ldots,\Mp\},\quad \up_j \in \setx',\ \dcp_K \subset \sety' \} $ with rate $R'_c= C(\text{g},P)- \epsilon,\ 0\leq \epsilon \leq C(\text{g},P)$ and a vanishing error probability $ \lambda'_n,\ \lim_{n \to \infty} \lambda'_n=0$. Every codeword $\up_K=(\up_{K,1},\up_{K,2},\ldots,\up_{K,n_{0}})$ fulfills the following energy constraint: \begin{equation}
%\sum_{l=1}^{n} {\up_{j,l}}^2 \leq n \cdot P  \end{equation}.
  $E_i(K)$ is encoded to a sequence $T^{\lceil \sqrt{n} \rceil}$ using a wiretap code as proposed in \cite{Idwiretapchannels}.
  The sequence $T^{m}$, obtained by concatenating $T^{n}$ and $T^{\lceil \sqrt{n} \rceil}$, where $m=n+\lceil \sqrt{n} \rceil $ as depicted in Fig. \ref{fig:codingscheme}, fulfills the power constraint:
 \begin{equation}
    \mathbb{E}[T_{i}^{2}]\leq P  \quad \forall i=1\dots m.
\nonumber \end{equation}
\begin{figure}[hbt!]
    \centering
    \input{figures/ID.tex}
    \caption{Coding scheme.}
    \label{fig:codingscheme}
\end{figure}
 $T^{m}$ is sent over the Wiretap channel.
  
  %$\bs{u^{\prime\prime}}$ of length $\sqrt{n}$ using a  wiretap-code $\Cpp=\{(\upp_k,\dcpp_k),k \in \{1,\ldots,\Mpp\}, \quad \upp_k \in \setx'',\  \dcpp_k\subset \sety'' \}$ is a wiretap code with rate $R''_c= \epsilon,\ 0\leq \epsilon \leq C(\text{g},P)$ and a vanishing error probability $ \lambda''_n,\ \lim_{n \to \infty} \lambda''_n=0$. Every codeword $\upp_k=(\upp_{k,1},\upp_{k,2},\ldots,\upp_{k,n})$ fulfills the following power constraint: \begin{equation}\sum_{l=1}^{n} {\upp_{k,l}}^2 \leq \lceil\sqrt{n} \rceil  \cdot P \end{equation}
 % $T^{n}$, obtained by concatenating $\bs{u^{\prime}}$ and $\bs{u^{\prime\prime}}$, is sent over the Gaussian Wiretap channel, where $Z_{1}^{n}$ is the output of the channel to Bob.
  Bob generates $L=\Psi(Y^{n},Z^{n})$ such that $\text{Pr}[K\neq L]$ is low.
  Since we have used a wiretap code in the second part, Bob, interested in $i^{\prime}$, can identify whether the message of interest was sent or not.
  We choose the rate of the first code to be approximately equal to the capacity of the channel to the legitimate receiver so that Bob can identify the message at a rate approximately equal to $C(g,P),$ the transmission capacity  of the main channel, without paying a price for the identification task.
  Although Eve can decode with low error probability the sequence  $T^{n}$, she cannot, with this setting, identify
the color, i.e., the second fundamental part of the sent codeword, even if she knows the correlation between the sources. Thus, the wiretapper cannot identify  the message  $i$. For more details regarding the proof, we refer the reader to \cite{wafapaper}.
Let us denote the CR capacity for this model by $C_{CR}(g,P).$ Then, by applying the Transformator-Lemma \cite{trafo}\cite{Idfeedback}, it holds that 
\begin{align}
    C_{SID}^{c}(g,g^{\prime},P) \geq C_{CR}(g,P)   \hspace{0.3cm} \text{if} \  C_{S}(g,g^{\prime},P)>0.
\nonumber \end{align}

In addition, it holds by Proposition \ref{proposition1} that
\begin{align}
C_{CR}(g,P)=\underset{\substack{U\\{\substack{U \circlearrow{X} \circlearrow{Y}\\ I(U;X)-I(U;Y) \leq C(g,P)}}}}{\max} I(U;X).
    \nonumber \end{align}
%Furthermore, it holds by the results in \cite{wafapaper} that:
%\begin{align}
%C_{SID}(g,g^{\prime},P)>0 \iff %C_{S}(g,g^{\prime},P)>0
%\end{align}
\end{proof}
 \begin{remark}
 It has recently been proven in \cite{wafapaper} that the secure identification capacity with randomized encoding is equal to the capacity of the channel to the legitimate receiver $C(g,P)$, provided that the secrecy capacity is strictly positive. As long as $P$ is chosen to satisfy $C(g,P)\leq H(X)$, the lower bound in Proposition $\ref{proposition3}$ may exceed the capacity of the main channel. Let us reconsider the example of binary sources presented in Example $\ref{example1}$. It is observed in Fig. $\ref{figplotcomparetolowerbound}$ that we can achieve a performance gain of at least $0.278$ for $P\approx1.72$, $\mu=0.2$, and $\sigma^2=1$.
\begin{figure}[hbt!]
  \centering
\definecolor{mycolor1}{rgb}{0.00000,0.44700,0.74100}%
\definecolor{mycolor2}{rgb}{1.00000,0.00000,1.00000}%
\scalebox{.64}{\begin{tikzpicture}

\begin{axis}[%
width=4.521in,
height=3.57in,
at={(0.758in,0.482in)},
scale only axis,
xmin=0,
xmax=3,
xlabel style={font=\color{white!15!black}},
xlabel={\textbf{Power}},
ymin=0,
ymax=1.01,
ylabel style={font=\color{white!15!black}},
ylabel={\textbf{Secure Identification Capacity}},
axis background/.style={fill=white},
xmajorgrids,
ymajorgrids,
legend style={at={(0.364,0.166)}, anchor=south west, legend cell align=left, align=left, draw=white!15!black}
]
\addplot [color=mycolor1]
  table[row sep=crcr]{%
0	0\\
0.01	0.00717764648853503\\
0.02	0.0142845760983855\\
0.03	0.0213221687042469\\
0.04	0.0282917641831838\\
0.05	0.035194663945699\\
0.06	0.0420321323942373\\
0.07	0.0488053983132112\\
0.08	0.055515656194372\\
0.09	0.0621640675011009\\
0.1	0.0687517618749675\\
0.11	0.0752798382876907\\
0.12	0.0817493661414398\\
0.13	0.0881613863202314\\
0.14	0.0945169121950086\\
0.15	0.100816930584825\\
0.16	0.107062402676424\\
0.17	0.11325426490434\\
0.18	0.119393429793558\\
0.19	0.125480786766609\\
0.2	0.131517202916897\\
0.21	0.137503523749935\\
0.22	0.143440573894081\\
0.23	0.149329157782258\\
0.24	0.155170060306075\\
0.25	0.160964047443681\\
0.26	0.166711866862596\\
0.27	0.172414248498721\\
0.28	0.178071905112638\\
0.29	0.183685532824265\\
0.3	0.189255811626865\\
0.31	0.194783405881363\\
0.32	0.200268964791864\\
0.33	0.205713122863233\\
0.34	0.211116500341524\\
0.35	0.216479703638053\\
0.36	0.221803325737807\\
0.37	0.227087946592901\\
0.38	0.232334133501722\\
0.39	0.237542441474391\\
0.4	0.242713413585121\\
0.41	0.247847581312035\\
0.42	0.252945464864979\\
0.43	0.258007573501832\\
0.44	0.263034405833794\\
0.45	0.268026450120105\\
0.46	0.272984184552646\\
0.47	0.27790807753082\\
0.48	0.282798587927113\\
0.49	0.287656165343718\\
0.5	0.292481250360578\\
0.51	0.297274274775177\\
0.52	0.30203566183443\\
0.53	0.306765826458964\\
0.54	0.311465175460088\\
0.55	0.316134107749756\\
0.56	0.320773014543762\\
0.57	0.325382279558451\\
0.58	0.329962279201189\\
0.59	0.334513382754815\\
0.6	0.339035952556319\\
0.61	0.343530344169946\\
0.62	0.34799690655495\\
0.63	0.352435982228176\\
0.64	0.35684790742168\\
0.65	0.361233012235545\\
0.66	0.3655916207861\\
0.67	0.369924051349664\\
0.68	0.374230616502018\\
0.69	0.37851162325373\\
0.7	0.382767373181489\\
0.71	0.386998162555587\\
0.72	0.391204282463687\\
0.73	0.395386018931\\
0.74	0.399543653037002\\
0.75	0.403677461028802\\
0.76	0.407787714431286\\
0.77	0.411874680154136\\
0.78	0.415938620595837\\
0.79	0.419979793744766\\
0.8	0.423998453277475\\
0.81	0.42799484865424\\
0.82	0.431969225211986\\
0.83	0.435921824254659\\
0.84	0.439852883141144\\
0.85	0.443762635370794\\
0.86	0.447651310666653\\
0.87	0.451519135056456\\
0.88	0.455366330951456\\
0.89	0.459193117223174\\
0.9	0.462999709278112\\
0.91	0.466786319130512\\
0.92	0.470553155473216\\
0.93	0.474300423746678\\
0.94	0.478028326206201\\
0.95	0.481737061987443\\
0.96	0.485426827170242\\
0.97	0.489097814840826\\
0.98	0.492750215152442\\
0.99	0.496384215384462\\
1	0.5\\
1.01	0.503597750702102\\
1.02	0.507177646488535\\
1.03	0.510739863705226\\
1.04	0.514284576098385\\
1.05	0.517811954865361\\
1.06	0.521322168704247\\
1.07	0.5248153838623\\
1.08	0.528291764183184\\
1.09	0.531751471153079\\
1.1	0.535194663945699\\
1.11	0.53862149946623\\
1.12	0.542032132394237\\
1.13	0.545426715225557\\
1.14	0.548805398313211\\
1.15	0.552168329907368\\
1.16	0.555515656194372\\
1.17	0.558847521334877\\
1.18	0.562164067501101\\
1.19	0.565465434913224\\
1.2	0.568751761874967\\
1.21	0.572023184808353\\
1.22	0.575279838287691\\
1.23	0.57852185507279\\
1.24	0.58174936614144\\
1.25	0.584962500721156\\
1.26	0.588161386320232\\
1.27	0.591346148758095\\
1.28	0.594516912195009\\
1.29	0.59767379916111\\
1.3	0.600816930584825\\
1.31	0.603946425820666\\
1.32	0.607062402676424\\
1.33	0.610164977439778\\
1.34	0.61325426490434\\
1.35	0.616330378395137\\
1.36	0.619393429793558\\
1.37	0.622443529561767\\
1.38	0.625480786766609\\
1.39	0.628505309103012\\
1.4	0.631517202916897\\
1.41	0.634516573227619\\
1.42	0.637503523749935\\
1.43	0.640478156915528\\
1.44	0.643440573894081\\
1.45	0.646390874613923\\
1.46	0.649329157782258\\
1.47	0.652255520904976\\
1.48	0.655170060306075\\
1.49	0.658072871146678\\
1.5	0.660964047443681\\
1.51	0.663843682088024\\
1.52	0.666711866862596\\
1.53	0.669568692459793\\
1.54	0.672414248498721\\
1.55	0.675248623542066\\
1.56	0.678071905112638\\
1.57	0.680884179709577\\
1.58	0.683685532824265\\
1.59	0.686476048955915\\
1.6	0.689255811626865\\
1.61	0.69202490339758\\
1.62	0.694783405881363\\
1.63	0.697531399758789\\
1.64	0.700268964791864\\
1.65	0.702996179837918\\
1.66	0.705713122863233\\
1.67	0.708419870956415\\
1.68	0.711116500341524\\
1.69	0.71380308639095\\
1.7	0.716479703638053\\
1.71	0.719146425789573\\
1.72	0.721803325737807\\
1.73	0.724450475572564\\
1.74	0.727087946592901\\
1.75	0.729715809318649\\
1.76	0.732334133501722\\
1.77	0.734942988137232\\
1.78	0.737542441474391\\
1.79	0.740132561027231\\
1.8	0.742713413585121\\
1.81	0.745285065223101\\
1.82	0.747847581312034\\
1.83	0.750401026528579\\
1.84	0.752945464864979\\
1.85	0.75548095963869\\
1.86	0.758007573501832\\
1.87	0.760525368450482\\
1.88	0.763034405833794\\
1.89	0.765534746362977\\
1.9	0.768026450120105\\
1.91	0.77050957656678\\
1.92	0.772984184552646\\
1.93	0.775450332323761\\
1.94	0.77790807753082\\
1.95	0.780357477237239\\
1.96	0.782798587927113\\
1.97	0.78523146551302\\
1.98	0.787656165343718\\
1.99	0.79007274221169\\
2	0.792481250360578\\
2.01	0.794881743492489\\
2.02	0.797274274775177\\
2.03	0.799658896849113\\
2.04	0.80203566183443\\
2.05	0.804404621337762\\
2.06	0.806765826458964\\
2.07	0.809119327797727\\
2.08	0.811465175460088\\
2.09	0.813803419064825\\
2.1	0.816134107749756\\
2.11	0.818457290177939\\
2.12	0.820773014543762\\
2.13	0.823081328578947\\
2.14	0.825382279558451\\
2.15	0.827675914306277\\
2.16	0.829962279201189\\
2.17	0.832241420182341\\
2.18	0.834513382754815\\
2.19	0.836778211995072\\
2.2	0.839035952556319\\
2.21	0.841286648673789\\
2.22	0.843530344169946\\
2.23	0.8457670824596\\
2.24	0.84799690655495\\
2.25	0.850219859070546\\
2.26	0.852435982228176\\
2.27	0.854645317861679\\
2.28	0.856847907421679\\
2.29	0.859043791980258\\
2.3	0.861233012235545\\
2.31	0.863415608516247\\
2.32	0.8655916207861\\
2.33	0.867761088648269\\
2.34	0.869924051349664\\
2.35	0.872080547785205\\
2.36	0.874230616502018\\
2.37	0.876374295703567\\
2.38	0.87851162325373\\
2.39	0.88064263668081\\
2.4	0.882767373181488\\
2.41	0.884885869624724\\
2.42	0.886998162555587\\
2.43	0.889104288199044\\
2.44	0.891204282463687\\
2.45	0.893298180945403\\
2.46	0.895386018931\\
2.47	0.897467831401768\\
2.48	0.899543653037002\\
2.49	0.901613518217464\\
2.5	0.903677461028802\\
2.51	0.905735515264918\\
2.52	0.907787714431286\\
2.53	0.909834091748228\\
2.54	0.911874680154136\\
2.55	0.91390951230866\\
2.56	0.915938620595837\\
2.57	0.917962037127187\\
2.58	0.919979793744766\\
2.59	0.921991922024163\\
2.6	0.923998453277475\\
2.61	0.925999418556223\\
2.62	0.92799484865424\\
2.63	0.929984774110513\\
2.64	0.931969225211986\\
2.65	0.933948231996327\\
2.66	0.935921824254659\\
2.67	0.937890031534244\\
2.68	0.939852883141144\\
2.69	0.941810408142836\\
2.7	0.943762635370794\\
2.71	0.945709593423039\\
2.72	0.947651310666653\\
2.73	0.949587815240257\\
2.74	0.951519135056456\\
2.75	0.953445297804259\\
2.76	0.955366330951456\\
2.77	0.95728226174697\\
2.78	0.959193117223174\\
2.79	0.961098924198184\\
2.8	0.962999709278112\\
2.81	0.964895498859299\\
2.82	0.966786319130512\\
2.83	0.968672196075116\\
2.84	0.970553155473216\\
2.85	0.97242922290377\\
2.86	0.974300423746678\\
2.87	0.976166783184843\\
2.88	0.978028326206202\\
2.89	0.979885077605734\\
2.9	0.981737061987443\\
2.91	0.983584303766314\\
2.92	0.985426827170242\\
2.93	0.987264656241941\\
2.94	0.989097814840826\\
2.95	0.99092632664487\\
2.96	0.992750215152442\\
2.97	0.994569503684117\\
2.98	0.996384215384462\\
2.99	0.998194373223811\\
3	1\\
};
\addlegendentry{\textbf{Capacity of secure identification} \\ \textbf{in presence of local randomness}}

\addplot [color=mycolor2, mark=square, mark options={solid, mycolor2}]
  table[row sep=crcr]{%
0	0\\
0.4	0.356715966801272\\
0.443700637652136	0.386983717545609\\
0.487401275304271	0.428427101978334\\
0.815858733228839	0.631816959696426\\
0.865858733228839	0.658549119279064\\
0.915858733228839	0.683883457664767\\
1.25	0.83559452771623\\
1.28947576084681	0.843763551865788\\
1.32895152169362	0.857478655511383\\
1.62047051030039	0.963165181222168\\
1.67047051030039	0.979631387926287\\
1.72047051030039	0.999999999999956\\
2	0.999999999999956\\
2.039600717839	0.999999999999956\\
2.079201435678	0.999999999999956\\
2.29304790375616	0.999999999999956\\
2.34304790375616	0.999999999999956\\
2.39304790375616	0.999999999999956\\
2.6	0.999999999999956\\
2.62530317679635	0.999999999999956\\
2.6506063535927	0.999999999999956\\
2.74211629784053	0.999999999999956\\
2.79211629784053	0.999999999999956\\
2.84211629784053	0.999999999999956\\
2.9	0.999999999999956\\
2.93006653339619	0.999999999999956\\
2.96013306679237	0.999999999999956\\
2.98006653339619	0.999999999999956\\
3	0.999999999999956\\
3.1	0.999999999999956\\
3.2	0.999999999999956\\
3.3	0.999999999999956\\
};
\addlegendentry{\textbf{Lower bound on CR-assisted} \\ \textbf{secure identification capacity}}

\end{axis}
\end{tikzpicture}}%
  \caption{Comparison of the lower bound on CR-assisted secure identification to the capacity of secure identification with randomized encoding for  a noise variance $\sigma^2=1$, for the correlated sources presented in Example \ref{example1} with $\mu=0.2.$}
  \label{figplotcomparetolowerbound}
\end{figure}
%The lower bound in Proposition \ref{proposition3} exceeds $C(g,P)$, the capacity of the main channel, for all $P$ satisfying $ C(g,P)<H(X)$.

Furthermore, if we consider identification over Gaussian channels without the availability of CR as a resource and with deterministic encoding, then, in this case, the identification capacity, measured on the log-log scale as already introduced in Definition $\ref{defcapacity}$, is equal to zero. This implies that the secure identification capacity for this communication scenario is also always equal to zero. This demonstrates, as already noted in Remark $\ref{remarkcorrelation}$, that in contrast to transmission, the resource CR allows for a significant performance gain in the identification task.
%This proves that, in contrast to transmission, correlation allows a significant performance gain in the identification task.
\end{remark}
 \begin{remark}
Clearly, we can proceed analogously to derive a lower bound for CR-assisted secure identification capacity over MIMO GWCs.
\end{remark}

 %We focus now on common randomness generation when we have a communication through single-user MIMO Gaussian Wiretap channels, as depicted in Fig. . The encoding procedure is done as follows:
 %Each codeword $\bs{T}^{n}$ sent over the Gaussian channel is obtained by concatenating two codewords as shown in Fig. The first part of the codeword has length $n_{1}$ and it contains the information about $K$, the random variable generated at $\mathcal{X}$. Both Bob and Eve are able to decode this part with small error probability. Then, Bob will generate the random variable $L$ such that $\text{Pr}[K\neq L]\leq \epsilon$.
 %The second part is some codeword  of a certain wiretap-code. It has length $\sqrt{n_{1}}$, where $n=n_{1}+\sqrt{n_{1}}.$ As result, only Bob is able to decode this part with small error probability.
 %we have generated common randomness between Alice and Bob. Now we can take advantage of the common randomness in the identification task, which is proved with this setting to be secure in \cite{wafapaper}.

 \section{Conclusions} \label{conclusion}
We studied the problem of CR generation over single-user SISO and MIMO Gaussian channels due to their practical relevance in various communication scenarios, such as satellite and deep space communication links, wired and wireless communications, etc. We provided a single-letter characterization of the CR capacity for both scenarios along with rigorous proofs. Additionally, we demonstrated that through CR generation, significant performance gains could be achieved in Post-Shannon communication tasks, which could be advantageous in numerous new applications, including machine-to-machine and human-to-machine systems, as well as the tactile internet. Specifically, we proposed a coding scheme for secure identification over the GWC with CR available as a resource and established a lower bound on the secure identification capacity within this framework. This lower bound may exceed the transmission capacity of the main channel, which is equal to the secure identification capacity in the case of randomized encoding, provided the secrecy capacity is strictly positive. As a future work, we suggest investigating the impact of antenna correlation on the CR capacity of MIMO Gaussian channels. Subsequent research could focus on providing a single-letter characterization of the CR-assisted secure identification capacity of the GWC and exploring CR-assisted identification for continuous-time channels.

\appendix
\section{Appendix}
 
 \subsection{Proof of the Existence of a random variable $V$  as defined in  \eqref{existV} }
 \begin{proof}
 We want to show that a random variable $V$ exists such that
$U_{\ell}\circlearrow{V}\circlearrow{X}\circlearrow{Y},\ \ell=1\hdots N_{\text{min}}$ with $I(V;X)=\sum_{\ell=1}^{N_{\text{min}}}I(U_{\ell};X)$ and $I(V;Y)=\sum_{\ell=1}^{N_{\text{min}}}I(U_{\ell};Y)$, where $U_{\ell}, \ell=1\hdots N_{\text{min}},$ satisfy the following constraints:
 \begin{enumerate}
     \item Each $U_{\ell}, \ \ell=1\hdots N_{\text{min}},$ is independent of $(X,Y).$
     \item $I(U_{\ell};X)-I(U_{\ell};Y) \leq C(\tilde{P}_{\ell}) \quad \ell=1\hdots N_{\text{min}}$
     \item $U_{\ell}$, $\ell=1\hdots N_{\text{min}}$, are pairwise independent
 \end{enumerate}
 It suffices to consider $V=U_{1}\hdots U_{N_{\text{min}}}$. Then it holds that
 \begin{align}
     I(V;X)&=I(U_{1}\hdots U_{N_{\text{min}}};X) \nonumber \\
           &=\sum_{\ell=1}^{N_{\text{min}}} I(U_{\ell};X|U_{1}\hdots U_{\ell-1}) \nonumber \\
           &\overset{(a)}{=}\sum_{\ell=1}^{N_{\text{min}}} I(U_{\ell};X),
 \nonumber \end{align}
 where $(a)$ follows because $U_{\ell},$ $\ell=1 \ldots N_{\text{min}},$ are pairwise independent and because each $U_{\ell},$ $\ell=1 \ldots N_{\text{min}},$ is independent of $X.$  Analogously, it holds that  $ I(V;Y)=\sum_{\ell=1}^{N_{\text{min}}}I(U_{\ell};Y)$. \\~\\
 The Markov chain $U_\ell \circlearrow{V}\circlearrow{X}\circlearrow{Y},$ $\ell=1\hdots N_{\text{min}},$ is satisfied since  for $\ell=1\hdots N_{\text{min}},$ we have
\begin{align}
    &\mbb P[Y=y|X=x,V=v,U_\ell=u_\ell] \nonumber \\
    &\overset{(a)}{=}\mbb P[Y=y|X=x,V=v] \nonumber \\
    &=\mbb P[Y=y|X=x,U=u_1,U_2=u_2,\ldots ,{U_N}_{\text{min}}={u_N}_{\text{min}}] \nonumber \\
   &\overset{(b)}{=} \mbb P[Y=y|X=x], 
\nonumber \end{align}
where $(a)$ follows because  $V=U_{1}\hdots U_{N_{\text{min}}}$ and where
$(b)$ follows because each $U_{\ell}$ is independent of $(X,Y).$ 
 \end{proof}
\subsection{Proof of the Non-Convexity of $g_{0}\left(\boldsymbol{\theta}\right)$}
\label{proofconcavity}
It holds that
\begin{align}
    g_{0}\left(\bs{\theta}\right) &=-I(U;X) \nonumber \\ &=-H(X)-H(U)+H(XU),
\nonumber \end{align}
which yields 
\begin{equation}
    \frac{\partial g_{0}\left(\bs{\theta}\right)}{\partial P_{UX}\left(u,x\right)}=\frac{\partial H(XU)}{\partial P_{UX}\left(u,x\right)}-\frac{\partial H(U)}{\partial P_{UX}\left(u,x\right)}.
\nonumber \end{equation}

On one side, we have
\begin{align}
    &H(U)\nonumber \\&=-\underset{u'\in\mathcal{U}}{\sum}P_{U}\left(u'\right)\log\left(P_{U}\left(u'\right)\right) \nonumber \\
    &=-\underset{u'\in\mathcal{U}}{\sum} \left(\underset{x' \in \mathcal{X}}{\sum} P_{UX}\left(u',x'\right)              \right) \log\left(\underset{x' \in \mathcal{X}}{\sum} P_{UX}\left(u',x'\right)              \right)
\nonumber \end{align}
yielding
\begin{equation}
    \frac{\partial H(U)}{\partial P_{UX}\left(u,x\right)}=-\left[\log\left(\underset{x'\in \mathcal{X}}{\sum}P_{UX}\left(u,x'\right)\right)+\frac{1}{\ln(2)}\right].
\nonumber \end{equation}

On the other side, we have
\begin{equation}
    H(UX)=-\underset{u'\in\mathcal{U},x'\in\mathcal{X}}{\sum}  P_{UX}\left(u',x'\right)\log\left( P_{UX}\left(u',x'\right)              \right)
\nonumber \end{equation}
yielding
\begin{equation}
    \frac{\partial H(UX)}{\partial P_{UX}\left(u,x\right)}=-\left[\log\left(P_{UX}\left(u,x\right)\right)+\frac{1}{\ln(2)}\right].
\nonumber \end{equation}
    Thus, we obtain
    \begin{equation}
         \frac{\partial g_{0}\left(\bs{\theta}\right)}{\partial P_{UX}\left(u,x\right)}=\log\left(\frac{\underset{x' \in \mathcal{X}}{\sum} P_{UX}\left(u,x'\right) }{P_{UX}\left(u,x\right)}
         \right).
         \label{derivg0}
    \end{equation}
    As a result,
    \begin{equation}
      \frac{\partial^{2}g_{0}\left(\bs{\theta}\right)}{\partial^{2} P_{UX}\left(u,x\right)} =\frac{P_{UX}\left(u,x\right)-\underset{x' \in \mathcal{X}}{\sum} P_{UX}\left(u,x'\right) }{P_{UX}\left(u,x\right) \underset{x' \in \mathcal{X}}{\sum} P_{UX}\left(u,x'\right)}\leq 0
    \nonumber \end{equation}
    where 
    \begin{equation}
        \underset{u' \in \mathcal{U},x' \in \mathcal{X}}{\sum} \frac{\partial^{2}g_{0}\left(\bs{\theta}\right)}{\partial^{2} P_{UX}\left(u',x'\right)} <0. \nonumber
    \end{equation}
    This implies that the Hessian matrix of $g_{0}$ is not positive semi-definite, which proves the non-convexity of $g_{0}$ in $\bs{\theta}$.
\subsection{Computation of $\frac{\partial g_{1}\left(\boldsymbol{\theta}\right)}{\partial P_{UX}\left(u,x\right)}$ for $U\circlearrow{X}\circlearrow{Y}$ }
It holds that
\begin{align}
    g_{1}\left(\boldsymbol{\theta}\right)&=I(U;X|Y) -\min\{C(P),H(X|Y)\}
    \nonumber \\
    &=H(X|Y)+H(U|Y)-H(UX|Y) -\min\{C(P),H(X|Y)\}.
\nonumber \end{align}
Thus, we have
\begin{equation}
    \frac{\partial g_{1}\left(\boldsymbol{\theta}\right)}{\partial P_{UX}\left(u,x\right)}=\frac{\partial H(U|Y)}{\partial P_{UX}\left(u,x\right)}-\frac{\partial H(UX|Y)}{\partial P_{UX}\left(u,x\right)}.
\nonumber \end{equation}
The Markov chain $U\circlearrow{X}\circlearrow{Y}$ implies that \\~\\ $\forall x \in \mathcal{X}$, $\forall u \in \mathcal{U}$ and $\forall y \in \mathcal{Y}$:
\begin{equation}
   P_{UX|Y}\left(u,x|y\right)=\frac{P_{Y|X}\left(y|x\right)P_{UX}\left(u,x\right)}{P_{Y}\left(y\right)}.
   \label{prop1}
\end{equation}
yielding
\begin{equation}
    \frac{\partial P_{UX|Y}\left(u,x|y\right)}{\partial P_{UX}\left(u,x\right)}=\frac{P_{Y|X}\left(y|x\right)}{P_{Y}\left(y\right)}
    \label{propderiv1}
\end{equation}
and that $\forall u \in \mathcal{U}$ and $\forall y \in \mathcal{Y}$
\begin{equation}
   P_{U|Y}\left(u|y\right)=\underset{x' \in \mathcal{X}}{\sum}\frac{P_{Y|X}\left(y|x'\right)P_{UX}\left(u,x'\right)}{P_{Y}\left(y\right)}, 
   \label{prop2}
\end{equation}
yielding
\begin{equation}
    \frac{\partial P_{U|Y}\left(u|y\right)}{\partial P_{UX}\left(u,x\right)}=\frac{P_{Y|X}\left(y|x\right)}{P_{Y}\left(y\right)}.
    \label{propderiv2}
\end{equation}
Furthermore, we have
\begin{align}
    &H(UX|Y) \nonumber \\
    &=-\underset{y\in\mathcal{Y}}{\sum}P_{Y}\left(y\right) \underset{u'\in\mathcal{U},x'\in\mathcal{X}}{\sum}P_{UX|Y}\left(u',x'|y\right)\log\left(P_{UX|Y}\left(u',x'|y\right)\right),
\nonumber \end{align}
which yields
\begin{align}
    &\frac{\partial H(UX|Y)}{\partial P_{UX}\left(u,x\right)}\nonumber \\
    &\overset{(a)}{=}-\underset{y \in \mathcal{Y}}{\sum}P_{Y|X}\left(y|x\right)\left[\log\left(P_{UX|Y}\left(u,x|y\right)\right)+\frac{1}{\ln(2)}\right],
\nonumber \end{align}
where $(a)$ follows from using the sum and product rule of derivatives and from $\eqref{propderiv1}$.
Similarly, it holds that
\begin{align}
    H(U|Y)=
    -\underset{y\in\mathcal{Y}}{\sum}P_{Y}\left(y\right)\underset{u'\in\mathcal{U}}{\sum}P_{U|Y}\left(u'|y\right)\log\left(P_{U|Y}\left(u'|y\right)\right).
    %&-\underset{y \in \mathcal{Y}}{\sum}P_{Y}\left(y\right)\underset{u \in \mathcal{U}}{\sum} \left(\underset{x \in \mathcal{X}}{\sum} \frac{P_{UX}\left(u,x\right)}{P_{X}\left(x\right)} P_{X|Y}\left(x|y\right)            \right) \log\left(\underset{x \in \mathcal{X}}{\sum} \frac{P_{UX}\left(u,x\right)}{P_{X}\left(x\right)} P_{X|Y}\left(x|y\right)            \right) 
\nonumber \end{align}
Thus, we obtain
\begin{align}
    &\frac{\partial H(U|Y)}{\partial P_{UX}\left(u,x\right)} \nonumber \\
    &\overset{(a)}{= }-\underset{y \in \mathcal{Y}}{\sum}P_{Y|X}\left(y|x\right)\left[\log\left(P_{U|Y}\left(u|y\right)\right)+\frac{1}{\ln(2)}\right],
\nonumber \end{align}
where $(a)$ follows from using the sum and product rule of derivatives and from $\eqref{propderiv2}$.

As a result, we have
\begin{align}
    &\frac{\partial g_{1}\left(\boldsymbol{\theta}\right)}{\partial P_{UX}\left(u,x\right)}\nonumber \\ &=\underset{y \in \mathcal{Y}}{\sum} P_{Y|X}\left(y|x\right) \log \left( \frac{P_{UX|Y}\left(u,x|y\right)}{P_{U|Y}\left(u|y\right)}\right) \nonumber \\
    &\overset{(a)}{= }\underset{y \in \mathcal{Y}}{\sum} P_{Y|X}\left(y|x\right) \log \left( \frac{P_{Y|X}\left(y|x\right)P_{UX}\left(u,x\right)}{\underset{x' \in \mathcal{X}}{\sum}P_{Y|X}\left(y|x'\right)P_{UX}\left(u,x'\right)}\right),
    \label{derivg1estimate}
\end{align}
where $(a)$ follows from $\eqref{prop1}$ and $\eqref{prop2}$.

\section*{Acknowledgments}
The authors acknowledge the financial support by the Federal Ministry of Education and Research of Germany in the program of “Souverän. Digital. Vernetzt.”. Joint project 6G-life, project identification number: 16KISK002.{\color{black}{ Holger Boche and Christian Deppe further gratefully acknowledge
the financial support by the BMBF Quantum Programm QD-CamNetz, Grant
16KISQ077, QuaPhySI, Grant 16KIS1598K, and QUIET, Grant 16KISQ093.}} Christian Deppe was supported  by the Bundesministerium 
f\"ur Bildung und Forschung (BMBF) through Grant 16KIS1005. Rami Ezzine and Wafa Labidi were supported by the BMBF through Grant 16KIS1003K.

% \ifCLASSOPTIONcaptionsoff
%   \newpage
% \fi



% http://www.michaelshell.org/tex/bibtex/
%\bibliographystyle{IEEEtran}
\bibliographystyle{IEEEtran}
\bibliography{IEEEabrv,confs-jrnls,refrences}

\begin{thebibliography}{10}
\providecommand{\url}[1]{#1}
\csname url@samestyle\endcsname
\providecommand{\newblock}{\relax}
\providecommand{\bibinfo}[2]{#2}
\providecommand{\BIBentrySTDinterwordspacing}{\spaceskip=0pt\relax}
\providecommand{\BIBentryALTinterwordstretchfactor}{4}
\providecommand{\BIBentryALTinterwordspacing}{\spaceskip=\fontdimen2\font plus
\BIBentryALTinterwordstretchfactor\fontdimen3\font minus \fontdimen4\font\relax}
\providecommand{\BIBforeignlanguage}[2]{{%
\expandafter\ifx\csname l@#1\endcsname\relax
\typeout{** WARNING: IEEEtran.bst: No hyphenation pattern has been}%
\typeout{** loaded for the language `#1'. Using the pattern for}%
\typeout{** the default language instead.}%
\else
\language=\csname l@#1\endcsname
\fi
#2}}
\providecommand{\BIBdecl}{\relax}
\BIBdecl

\bibitem{part2}
R.~{Ahlswede} and I.~{Csiszár}, ``Common randomness in information theory and cryptography. \uppercase{II. CR} capacity,'' \emph{IEEE Transactions on Information Theory}, vol.~44, no.~1, pp. 225--240, 1998.

\bibitem{corrsurvey}
M.~{Sudan}, H.~{Tyagi}, and S.~{Watanabe}, ``Communication for generating correlation: A unifying survey,'' \emph{IEEE Transactions on Information Theory}, vol.~66, no.~1, pp. 5--37, 2020.

\bibitem{Naor2020}
M.~Naor, M.~Parte, and E.~Yogev, ``The power of distributed verifiers in interactive proofs,'' in \emph{Proceedings of the Fourteenth Annual ACM-SIAM Symposium on Discrete Algorithms}.\hskip 1em plus 0.5em minus 0.4em\relax SIAM, 2020, pp. 1096--115.

\bibitem{Idchannels}
R.~{Ahlswede} and G.~{Dueck}, ``Identification via channels,'' \emph{IEEE Transactions on Information Theory}, vol.~35, no.~1, pp. 15--29, 1989.

\bibitem{shannon}
C.~E. Shannon, ``A mathematical theory of communication,'' \emph{Bell System Technical Journal}, vol.~27, pp. 379--423, 623--656, 1948.

\bibitem{application}
H.~{Boche} and C.~{Deppe}, ``Secure identification for wiretap channels; robustness, super-additivity and continuity,'' \emph{IEEE Transactions on Information Forensics and Security}, vol.~13, no.~7, pp. 1641--1655, 2018.

\bibitem{industrie40}
Y.~Lu, ``Industry 4.0: A survey on technologies, applications and open research issues,'' \emph{Journal of Industrial Information Integration}, vol.~6, no.~1, pp. 1--10, 2017.

\bibitem{6Gcomm}
G.~P. Fettweis and H.~Boche, ``6\uppercase{G}: The personal tactile internet—and open questions for information theory,'' \emph{IEEE BITS the Information Theory Magazine}, vol.~1, no.~1, pp. 71--82, 2021.

\bibitem{6Gpostshannon}
J.~Cabrera, H.~Boche, C.~Deppe, R.~Schaefer, C.~Scheunert, and F.~Fitzek, ``6\uppercase{G} and the post-shannon theory,'' in \emph{2022}, E.~Bertin, N.~Crespi, and T.~Magedanz, Eds., 2022.

\bibitem{Moulin}
P.~Moulin, ``The role of information theory in watermarking and its application to image watermarking,'' \emph{Signal Processing}, vol.~81, no.~6, pp. 1121--1139, 2001.

\bibitem{watermarkingahlswede}
R.~Ahlswede and N.~Cai, ``Watermarking identification codes with related topics on common randomness,'' in \emph{Berlin, Heidelberg: Springer Berlin Heidelberg}, 2006, pp. 107--153.

\bibitem{watermarking}
Y.~Steinberg and N.~Merhav, ``Identification in the presence of side information with application to watermarking,'' \emph{IEEE Transactions on Information Theory}, vol.~47, no.~4, pp. 1410--1422, 2001.

\bibitem{part1}
R.~{Ahlswede} and I.~{Csiszár}, ``Common randomness in information theory and cryptography. \uppercase{I}. secret sharing,'' \emph{IEEE Transactions on Information Theory}, vol.~39, no.~4, pp. 1121--1132, 1993.

\bibitem{Wifi}
S.~Mathur, W.~Trappe, N.~Mandayam, C.~Ye, and A.~Reznik, ``Radio-telepathy: Extracting a secret key from an unauthenticated wireless channel,'' in \emph{Proceedings of the 14th ACM International Conference on Mobile Computing and Networking}, ser. MobiCom '08, 2008, p. 128–139.

\bibitem{semanticwiretap}
M.~Bellare, S.~Tessaro, and A.~Vardy, ``Semantic security for the wiretap channel,'' in \emph{Advances in Cryptology – CRYPTO 2012}, ser. Lecture Notes in Computer Science, R.~Safavi-Naini and R.~Canetti, Eds.\hskip 1em plus 0.5em minus 0.4em\relax Berlin, Heidelberg: Springer, 2012, vol. 7417, pp. 294--311.

\bibitem{semanticsecurity}
M.~Wiese and H.~Boche, ``Semantic security via seeded modular coding schemes and ramanujan graphs,'' \emph{IEEE Transactions on Information Theory}, vol.~67, no.~1, pp. 52--80, 2021.

\bibitem{6Gandtrustworthiness}
G.~P. Fettweis and H.~Boche, ``On 6\uppercase{G} and trustworthiness,'' \emph{Communications of the ACM}, vol.~65, no.~4, pp. 48--49, 2022.

\bibitem{jammerref}
H.~Boche, R.~F. Schaefer, and H.~V. Poor, ``Denial-of-service attacks on communication systems: Detectability and jammer knowledge,'' \emph{IEEE Transactions on Signal Processing}, vol.~68, pp. 3754--3768, 2020.

\bibitem{trustworthiness}
Y.~Chen, T.~Oechtering, H.~Boche, M.~Skoglund, and Y.~Luo, ``Distribution-preserving integrated sensing and communication with secure reconstruction,'' in \emph{IEEE International Symposium on Information Theory}, 2024.

\bibitem{6Gperspective}
P.~Schwenteck, G.~T. Nguyen, H.~Boche, W.~Kellerer, and F.~H.~P. Fitzek, ``6\uppercase{G} perspective of mobile network operators, manufacturers, and verticals,'' \emph{IEEE Networking Letters}, vol.~5, no.~3, pp. 169--172, 2023.

\bibitem{mobile}
D.~D.~N. {Bevan}, V.~T. {Ermolayev}, A.~G. {Flaksman}, I.~M. {Averin}, and P.~M. {Grant}, ``Gaussian channel model for macrocellular mobile propagation,'' in \emph{2005 13th European Signal Processing Conference}, 2005, pp. 1--4.

\bibitem{Boche2019SecureIU}
H.~{Boche} and C.~{Deppe}, ``Secure identification under passive eavesdroppers and active jamming attacks,'' \emph{IEEE Transactions on Information Forensics and Security}, vol.~14, no.~2, pp. 472--485, 2019.

\bibitem{IDforDataStorage}
S.~{Baur}, C.~{Deppe}, and H.~{Boche}, ``Secure storage for identification; random resources and privacy leakage,'' \emph{IEEE Transactions on Information Forensics and Security}, vol.~14, no.~8, pp. 2013--2027, 2019.

\bibitem{sectransmission}
I.~{Csiszár} and J.~{Korner}, ``Broadcast channels with confidential messages,'' \emph{IEEE Transactions on Information Theory}, vol.~24, no.~3, pp. 339--348, 1978.

\bibitem{wafapaper}
W.~{Labidi}, C.~{Deppe}, and H.~{Boche}, ``Secure identification for \uppercase{G}aussian channels,'' in \emph{ICASSP 2020 - 2020 IEEE International Conference on Acoustics, Speech and Signal Processing (ICASSP)}, 2020, pp. 2872--2876.

\bibitem{Wyner}
A.~D. {Wyner}, ``The wire-tap channel,'' \emph{The Bell System Technical Journal}, vol.~54, no.~8, pp. 1355--1387, 1975.

\bibitem{explicitconstruction}
S.~{Verdu} and V.~K. {Wei}, ``Explicit construction of optimal constant-weight codes for identification via channels,'' \emph{IEEE Transactions on Information Theory}, vol.~39, no.~1, pp. 30--36, 1993.

\bibitem{implementation}
S.~Derebeyoğlu, C.~Deppe, and R.~Ferrara, ``Performance analysis of identification codes,'' \emph{Entropy}, vol.~22, no.~10, p. 1067, 2020.

\bibitem{Igor}
I.~{Bjelaković}, H.~{Boche}, and J.~{Sommerfeld}, ``Capacity results for compound wiretap channels,'' in \emph{2011 IEEE Information Theory Workshop}, 2011, pp. 60--64.

\bibitem{codingtheorems}
I.~Csiszár and J.~Körner, \emph{Information Theory: Coding Theorems for Discrete Memoryless Systems}, 2nd~ed.\hskip 1em plus 0.5em minus 0.4em\relax Cambridge University Press, 2011.

\bibitem{twoplayermethod}
A.~Cotter, H.~Jiang, and K.~Sridharan, ``Two-player games for efficient non-convex constrained optimization,'' in \emph{Proceedings of Machine Learning Research}, vol.~98, 2019, pp. 1--33.

\bibitem{AdaGrad}
J.~Duchi, E.~Hazan, and Y.~Singer, ``Adaptive subgradient methods for online learning and stochastic optimization,'' \emph{Journal of Machine Learning Research}, vol.~12, no.~61, pp. 2121--2159, 2011.

\bibitem{projectionsimplex}
W.~Wang and M.~A. Carreira-Perpin{\'a}n, ``Projection onto the probability simplex: An efficient algorithm with a simple proof, and an application,'' \emph{arXiv preprint arXiv:1309.1541}, 2013.

\bibitem{teletar}
E.~Telatar, ``Capacity of multi-antenna \uppercase{G}aussian channels,'' \emph{European Transactions on Telecommunications}, vol.~10, no.~6, pp. 585--595, 1999.

\bibitem{Tse}
D.~Tse and P.~Viswanath, \emph{Fundamentals of Wireless Communication}.\hskip 1em plus 0.5em minus 0.4em\relax New York, NY, USA: Cambridge University Press, 2005.

\bibitem{waterfilling}
W.~Yu, W.~Rhee, S.~Boyd, and J.~Cioffi, ``Iterative water-filling for \uppercase{G}aussian vector multiple-access channels,'' \emph{IEEE Transactions on Information Theory}, vol.~50, no.~1, pp. 145--152, 2004.

\bibitem{Idwiretapchannels}
R.~{Ahlswede} and Z.~{Zhang}, ``New directions in the theory of identification via channels,'' \emph{IEEE Transactions on Information Theory}, vol.~41, no.~4, pp. 1040--1050, 1995.

\bibitem{trafo}
R.~Ahlswede, ``General theory of information transfer: Updated,'' \emph{Discrete Applied Mathematics}, vol. 156, pp. 1348--1388, 2008.

\bibitem{Idfeedback}
R.~{Ahlswede} and G.~{Dueck}, ``Identification in the presence of feedback-a discovery of new capacity formulas,'' \emph{IEEE Transactions on Information Theory}, vol.~35, no.~1, pp. 30--36, 1989.

\end{thebibliography}

\end{document}